\documentclass[12pt]{article}
\usepackage{t1enc}
\usepackage{amsmath}
\usepackage{amsthm}
\usepackage{latexsym}
\usepackage{natbib}
\usepackage{epsfig}
\usepackage{textcomp}
\usepackage{alltt}
\usepackage{graphicx}
\usepackage{rotating}
\usepackage{dcolumn}
\usepackage{enumitem}
\usepackage{amsfonts}
\usepackage{algorithm}
\usepackage{amssymb}

\begin{document}
\newcommand{\cip}{\perp\!\!\!\perp}
\newcommand{\nothere}[1]{}
\newcommand{\noi}{\noindent}
\newcommand{\mbf}[1]{\mbox{\boldmath $#1$}}
\newcommand{\cond}{\, |\,}
\newcommand{\hO}[2]{{\cal O}_{#1}^{#2}}
\newcommand{\hF}[2]{{\cal F}_{#1}^{#2}}
\newcommand{\tl}[1]{\tilde{\lambda}_{#1}^T}
\newcommand{\la}[2]{\lambda_{#1}^T(Z^{#2})}
\newcommand{\I}[1]{1_{(#1)}}
\newcommand{\cd}{\mbox{$\stackrel{\mbox{\tiny{\cal D}}}{\rightarrow}$}}
\newcommand{\cp}{\mbox{$\stackrel{\mbox{\tiny{p}}}{\rightarrow}$}}
\newcommand{\cas}{\mbox{$\stackrel{\mbox{\tiny{a.s.}}}{\rightarrow}$}}
\newcommand{\ld}{\mbox{$\; \stackrel{\mbox{\tiny{def}}}{=3D} \; $}}
\newcommand{\nk}{\mbox{$n \rightarrow \infty$}}
\newcommand{\con}{\mbox{$\rightarrow $}}
\newcommand{\dprime}{\mbox{$\prime \vspace{-1 mm} \prime$}}
\newcommand{\Borel}{\mbox{${\cal B}$}}
\newcommand{\bevis}{\mbox{$\underline{\em{Proof}}$}}
\newcommand{\Rd}[1]{\mbox{${\Re^{#1}}$}}
\newcommand{\il}[1]{{\int_{0}^{#1}}}
\newcommand{\pl}[1]{\mbox{\bf {\LARGE #1}}}
\newcommand{\expit}{\text{expit}}
\newcommand{\indep}{\rotatebox[origin=c]{90}{$\models$}}
\newcommand{\blind}{1}
\newcommand{\pr}{\text{pr}}
\newcommand{\var}{\text{var}}
\newcommand{\cov}{\text{cov}}
\newcommand{\Bin}{\text{Bin}}
\newcommand{\Exp}{\text{Exp}}
\newcommand{\unif}{\text{unif}}
\newcommand{\logit}{\text{logit}}
\newcommand{\sign}{\text{sign}}
\newcommand{\support}{\text{support}}
\newcommand\norm[1]{\left\lVert#1\right\rVert}

\newtheorem{theorem}{Theorem}
\newtheorem{lemma}{Lemma}
\newtheorem{prop}{Proposition}
\newtheorem{assumption}{Assumption}
\newtheorem{definition}{Definition}
\newtheorem*{remark}{Remark}
\newtheorem{corollary}{Corollary}
\newtheorem{example}{Example}

\parindent12pt

\begin{center}{\Large{Doubly robust tests of exposure effects\\ under high-dimensional confounding}}
 \end{center}
 
 { 
\begin{center}
Oliver Dukes,\\
\textit{Department of Applied Mathematics, Computer Science and Statistics}\\
\textit{Ghent University, Belgium}\\[2ex]
Vahe Avagyan\\
\textit{Mathematical and Statistical Methods Group}\\
\textit{Wageningen University and Research, The Netherlands}\\[2ex]
and Stijn Vansteelandt\\
\textit{Department of Applied Mathematics, Computer Science and Statistics}\\
\textit{Ghent University, Belgium}\\
\textit{and Department of Medical Statistics}\\
\textit{London School of Hygiene and Tropical Medicine, U.K.}\\[2ex]
email: oliver.dukes@ugent.be
\end{center}
}

\bigskip \setlength{\parindent}{0.3in} \setlength{\baselineskip}{24pt}

\begin{abstract}
After variable selection, standard inferential procedures for regression parameters may not be \textit{uniformly valid}; there is no finite-sample size at which a standard test is guaranteed to approximately attain its nominal size. This problem is exacerbated in high-dimensional settings, where variable selection becomes unavoidable. This has prompted a flurry of activity in developing uniformly valid hypothesis tests for a low-dimensional regression parameter (e.g. the causal effect of an exposure $A$ on an outcome $Y$) in high-dimensional models. So far there has been limited focus on model misspecification, although this is inevitable in high-dimensional settings. We propose tests of the null that are uniformly valid under sparsity conditions weaker than those typically invoked in the literature, assuming working models for the exposure and outcome are both correctly specified. When one of the models is misspecified, by amending the procedure for estimating the nuisance parameters, our tests continue to be valid; hence they are \textit{doubly robust}. Our proposals are straightforward to implement using existing software for penalized maximum likelihood estimation and do not require sample-splitting. We illustrate them in simulations and an analysis of data obtained from the Ghent University Intensive Care Unit.
\end{abstract}
\noi

\section{Introduction}\label{introduction}

We will consider a study design which collects i.i.d. data on an outcome $Y$, an exposure of interest $A$ and a vector of covariates $L$, some of which may confound the relationship between $A$ and $Y$. A common means of assessing the effect of $A$ on $Y$ is to fit a regression model, adjusted for $A$ and the covariates; the estimate of the coefficient for $A$ is then used to obtain inference on the exposure effect. In practice, there is often little prior knowledge on which variables in a given data set are confounders, and furthermore how one should model the association between these confounders and outcome. Hence, data-adaptive procedures are typically employed in order to select the variables to adjust for and/or choose a model for their dependence on $Y$.  In particular, data-adaptive model selection becomes increasingly necessary when the dimension of $L$ is close to or greater than the number of observations. 

However, obtaining hypothesis tests and confidence intervals that approximately enjoy their nominal size/coverage after model selection is challenging. The estimate of the effect of $A$ obtained directly via regularization techniques - e.g. using a penalized maximum likelihood estimator (PMLE) - will inherit a so-called regularization bias. Furthermore, the moderate-sample distribution of this estimator will typically be non-normal \citep{leeb_model_2005}. This is because convergence to the asymptotic normal distribution is not \textit{uniform} with respect to the parameters indexing the model for $Y$. Therefore, there exists no finite $n$ such that normal-based tests and intervals are guaranteed to perform well. This issue applies more generally to post-regularization estimators (where the model selected via regularization is refitted using the chosen covariates) and routinely-used stepwise variable selection strategies. Standard inferential procedures also ignore the additional uncertainty created during the model-selection process.

In the mathematical statistics literature, this has prompted the development of methods to obtain \textit{uniformly valid} inference for a low-dimensional regression parameter in a high-dimensional model. Initial focus was given to tests and confidence intervals for a coefficient in a regression model fit using the Lasso \citep{belloni_inference_2014,van_de_geer_asymptotically_2014,zhang_confidence_2014}; after which attention has turned to more general data-adaptive methods \citep{ning_general_2017,chernozhukov_double/debiased_2018}. The key insight has been that one should perform selection based on an additional working model for the association between $A$ and $L$ (in addition to $Y$ with $A$ and $L$). 
The majority of the recent proposals rely on strong assumptions on sparsity e.g. the number of relevant covariates in $L$, which needs to be much smaller than the square root of the sample size $n$. In biostatistics, there were earlier developments in Targeted Maximum Likelihood Estimation (TMLE), where data-adaptive methods are incorporated into the estimation of causal effects \citep{van_der_laan_targeted_2011}. Much of the theory on TMLE is developed under Donsker conditions from the empirical process literature (but not all of it e.g. \citet{van_der_laan_cross-validated_2011}). These conditions are usually too restrictive in settings where the dimension of covariates is allowed to grow with sample size.

In this work, we describe how to obtain uniformly valid tests of the causal null hypothesis for a regression parameter in a high-dimensional Generalized Linear Model (GLM). Our tests require postulation of working models for the conditional mean of the exposure and the outcome given covariates. We will work under parametric models, as this is what is typically done in practice. First, we describe a procedure for estimating the nuisance parameters which yields a valid test so long as all working models are correct. However, given that regularization/model selection is required because we do not know the true models to start with, some degree of misspecification is likely. This is felt most acutely when the number of covariates in the data set is very large relative to the number of observations. We then show how to amend the earlier procedure for nuisance parameter estimation, so that the test statistic will converge uniformly to a limiting normal distribution if either working model is correct. Hence the test can be made \textit{uniformly doubly robust}. This is in contrast to several existing proposals, which give doubly robust estimators but not inference \citep{van_der_laan_targeted_2011,farrell_robust_2015,chernozhukov_double/debiased_2018,shah_hardness_2019}. Furthermore, we will show that when both working models are correct, then in certain cases the test will continue to attain its nominal size under sparsity conditions weaker than those invoked in the literature. Our test statistic is straightforward to construct, and all procedures for estimating the nuisance parameters can be performed using existing penalized regression software. Sample-splitting is not required, which makes our proposal much simpler to implement and less affected by regularization bias in moderate sample sizes.

The paper is organized as follows: in Section \ref{motivation}, we state the null hypothesis we are interested in testing and describe issues with obtaining valid inference in the high-dimensional setting. Section \ref{test} presents the score test statistic. In Sections \ref{nuis_cor} and \ref{nuis_mis}  we describe specific procedures for estimating the nuisance parameters, first when all working models are correct and then under misspecification. We also discuss the asymptotic properties of the various methods. 
We illustrate the methods via simulation studies in Section \ref{simulations} and an analysis of data from the Ghent University Intensive Care Unit in Section \ref{data_analysis}, where we consider the effect of a change in glycemia level on mortality in critically ill patients. 


\section{Motivation}\label{motivation}

We consider a test of the null hypothesis $H_0$ that $Y$ is independent of $A$ within strata defined by $L$, or
\begin{align}\label{cond_ind}
H_0: \quad Y\indep A|L.
\end{align}
We will let the exposure $A$ be binary e.g. it is coded as 1 if an individual undergoes a particular medical treatment and 0 otherwise; extensions to more general exposures will be discussed later in the paper. Then given the standard structural conditions in the causal inference literature, in particular that $L$ is sufficient to adjust for confounding, the null hypothesis also expresses the absence of a causal effect of $A$ on $Y$, conditional on $L$.

Simple methods for testing the null hypothesis of conditional independence are available via the regression framework. Standard score tests of $H_0$ require estimation of $E(Y|L)$, since under the null, $E(Y|A=a,L)=E(Y|L)$ (full conditional independence implies mean conditional independence). In realistic settings where $L$ has multiple continuous components, non-parametric estimators of this functional may perform poorly. A common strategy is to instead postulate a parametric regression model $\mathcal{B}$ for the mean of $Y$ conditional on the covariates:
\[E(Y|L)=m(L;\beta),\]
where $m(L;\beta)$ is a known function smooth in an unknown finite-dimensional parameter $\beta$.  
Then, via maximum likelihood estimation of GLMs, under the pre-specified model $\mathcal{B}$ one can obtain a consistent and uniformly asymptotically normal (UAN) score statistic for testing $H_0$ (where uniformity is with respect to $\beta$). This means that there exists a finite-sample size such that for any value of $\beta$ within the parameter space, the test statistic will be approximately normally distributed. We note that although this is a valid test of $H_0$, it is only consistent in the direction of alternatives that obey some form of mean conditional independence between $Y$ and $A$; the same holds for subsequent tests described in this paper.   

Unfortunately, this standard methodology does not straightforwardly extend to high-dimensional settings. 
In low-dimensional settings, one can perform a score test of the causal null $H_0$ based on the asymptotic distribution of an unbiased (unscaled) score test statistic $U(\beta)$; for likelihood estimation of canonical GLMs, $U(\beta)=A\{Y-m(L;\beta)\}$. Then let $\tilde{\beta}$ denote an estimate of $\beta$ obtained either directly via some regularization method or after model selection. Following a Taylor expansion,
\begin{align}\label{Mest_TE}
\frac{1}{\sqrt{n}}\sum^n_{i=1}U_i(\tilde{\beta})=&\frac{1}{\sqrt{n}}\sum^n_{i=1}U_i(\beta)+\frac{1}{n}\sum^n_{i=1}\frac{\partial U_i(\beta)}{\partial \beta}  \sqrt{n}(\tilde{\beta}-\beta)\nonumber\\&+O_{P}(\sqrt{n}||\tilde{\beta}-\beta||^2_2)
\end{align}
where $||.||_2$ denotes the Euclidean norm. For fixed $\beta$, by appealing to the oracle properties of $\tilde{\beta}$ it may be argued that the right hand side of (\ref{Mest_TE}) is asymptotically normal. Indeed,  assuming that $\tilde{\beta}$ converges sufficiently quickly, the remainder term $O_{P}(\sqrt{n}||\tilde{\beta}-\beta||^2_2)$ converges to zero.  But in general this does not prevent the existence of converging sequences $\beta_n$ for which $\sqrt{n}(\tilde{\beta}-\beta_n)$ and thus the test statistic has a complex, non-normal distribution. One root cause of this is the discrete nature of many data-adaptive methods e.g. stepwise selection; in some samples $\tilde{\beta}$ will be forced to zero whereas in others it will be allowed to take on its estimated value. This discrete behavior persists with increasing sample size under certain sequences $\beta_n$. The convergence of the resulting score test statistic to the limiting standard normal  is hence not uniform over the parameter space \citep{leeb_model_2005,dukes_how_2019}. This is troubling, as one wishes there to be a finite $n$ where the normal approximation is guaranteed to hold well, regardless of the (unknown) true values of the nuisance parameters, in order to guarantee that the procedure will work well in finite samples.


\section{A uniformly valid test of the causal null hypothesis}\label{test}


We introduce in this section the statistic we will use for testing $H_0$ in a high-dimensional setting. 
Let us now formally define $\gamma$ to be the nuisance parameter indexing the parametric  model $\mathcal{A}$ for the conditional mean of $A$ given $L$:
\[E(A|L)=\pi(L;\gamma),\] where $\pi(L;\gamma)$ is a known function smooth in an unknown finite-dimensional parameter $\gamma$; the conditional mean $E(A|L)$ is known as the \textit{propensity score} for binary $A$. Our analysis is based on the score function 
\[U(\eta)\equiv\{A-\pi(L;\gamma)\}\{Y-m(L;\beta)\}.\]
where $\eta=(\gamma^T,\beta^T)^T$. This will require initial estimates of $\gamma$ and $\beta$ under working models $\mathcal{A}$ and $\mathcal{B}$ respectively. It is natural to model the dependence of $A$ on $L$ using a logistic regression e.g. $\pi(L;\gamma)=\expit(\gamma^TL)$; if $A$ were continuous, one might postulate a linear or log-linear model instead and the proposal can then be easily adapted. The form that model $\mathcal{B}$ takes will depend on the nature of the outcome. If $Y$ is continuous and unconstrained, one might postulate a linear model e.g. $m(L;\beta)=\beta^TL$. 

One can then construct a test statistic
\[T_n=\frac{\frac{1}{\sqrt{n}}\sum^n_{i=1}U_i(\hat{\eta})}
{\sqrt{\frac{1}{n}\sum^n_{i=1}\big\{U_i(\hat{\eta})-\bar{U}(\hat{\eta})\big\}^2}}\]
that we will compare to the standard normal distribution. Here, $\bar{U}(\hat{\eta})=n^{-1}\sum^n_{i=1}U_i(\hat{\eta})$ and $\hat{\eta}$ is an estimate of $\eta$; in what follows, we will focus on regularized estimation of this parameter. Note that in evaluating the functions $\pi(L;\gamma)$ and $m(L;\beta)$ at their limiting values, it follows that the mean of $U(\eta)$ under the null is equal to 
\begin{align*}
E[\{E(A|L)-\pi(L;\gamma)\}\{E(Y|L)-m(L;\beta)\}]
\end{align*}
which equals zero if either model $\mathcal{A}$ or $\mathcal{B}$ is correct a.k.a under the union model  $\mathcal{A}\cup\mathcal{B}$. Hence we refer to the score $U(\eta)$ as doubly robust. 

\section{Estimation of $\eta$ when all models are correct}\label{nuis_cor}

\subsection{Proposal}

We will first consider data generating processes where both models $\mathcal{A}$ or model $\mathcal{B}$ are correctly specified e.g. we will work under the intersection submodel $\mathcal{A}\cap\mathcal{B}$. In high-dimensional settings, under model $\mathcal{A}\cap\mathcal{B}$ one can estimate $\gamma$ and $\beta$ as $\hat{\gamma}$ and $\hat{\beta}$ respectively using any sufficiently fast-converging sparse estimator. 
For example, with binary $A$ and continuous $Y$ and using standard PMLE with a Lasso penalty, $\hat{\gamma}$ and $\hat{\beta}$ can be obtained as:
\begin{align*}
\hat{\gamma}&=\arg\min_{\gamma}\frac{1}{n}\sum^{n}_{i=1}[\log\{1+\exp(\gamma^TL_i)\}-A_i(\gamma^TL_i)]+\lambda_\gamma||\gamma||_1\\
\hat{\beta}&=\arg\min_{\beta}\frac{1}{2n}\sum^{n}_{i=1}(Y_i-\beta^TL_i)^2+\lambda_\beta||\beta||_1
\end{align*}
\citep{tibshirani_regression_1996}, where $\lambda_\gamma>0$ and $\lambda_\beta>0$  are penalty parameters. To keep the notation simple, we have omitted the dependence of $\hat{\gamma}$ and $\hat{\beta}$ on $\lambda_{\gamma}$ and $\lambda_{\beta}$ respectively (as well as on $n$). Plugging the resulting estimates into $T_n$ will yield a test statistic that under the null follows a standard normal distribution. We note that our Theorem \ref{theorem1_drt} is quite general and covers other penalized $m$-estimation methods (including variants on the Lasso).

To give some intuition about why one can plug  $\hat{\gamma}$ and $\hat{\beta}$ into $T_n$ and yet obtain a UAN test statistic, by repeating the expansion in (\ref{Mest_TE}) for the score $U(\hat{\eta})$, we observe the first-order term
\begin{align}\label{problem}
\frac{1}{n}\sum^n_{i=1}\frac{\partial U_i(\eta)}{\partial \eta}  \sqrt{n}(\hat{\eta}-\eta).
\end{align}
We can control this term under the model $\mathcal{A}\cap\mathcal{B}$, since $\partial U(\eta)/\partial \eta$ will then have expectation zero at the limiting value of $\eta$. For example, when $\pi(L;\gamma)=\expit(\gamma^TL)$ and $m(L;\beta)=\beta^TL$, then using the law of iterated expectation
\begin{align*}
&E\{\partial U(\eta)/\partial \beta\}=E[E\{A-\pi(L;\gamma)|L\}L]=0 \quad \textrm{and}\\
&E\{\partial U(\eta)/\partial \gamma\}=E\left[\pi(L;\gamma)\{1-\pi(L;\gamma)\}E\{Y-m(L;\beta)|L\}L\right]=0
\end{align*}
under the null. This property helps to ensure that term (\ref{problem}) is asymptotically negligible, regardless of the complex behavior of $\hat{\eta}$ . Such phenomena (in the context of doubly robust estimators) is well-understood when $L$ is low-dimensional \citep{vermeulen_bias-reduced_2015}. What is surprising is that it continues to hold in high-dimensional settings, even when non-regular estimators are used for $\eta$ \citep{belloni_post-selection_2016}. 

We recommend selecting the penalty parameters in practice via cross validation, although there are limited theoretical results available on its validity in this context \citep{chetverikov_cross-validated_2016}, and our inferences assume that these parameters are fixed. The standard conditions are that $\lambda_\gamma=o(\sqrt{\log(p\lor n)/n})$ and $\lambda_\beta=o(\sqrt{\log(p\lor n)/n})$ (where $a\lor b$ denotes the maximum of $a$ and $b$), which are required for our theoretical results (see Appendix A). 
In practice, we also recommend refitting both working models; model refitting is typically done in the literature in order to improve finite-sample performance \citep{belloni_post-selection_2016,ning_general_2017}. Our theory can allow for this by appealing to results on Post-Lasso estimators \citep{belloni_inference_2014,belloni_post-selection_2016}. For any vector $a\in\mathbb{R}^{p}$, let us define its support as $\support(a)=\{j\in\{1,..,p\}:a_j\ne 0\}$; then we refit each model $\mathcal{A}$ and $\mathcal{B}$ using $\support(\gamma)\subseteq \support(\hat{\gamma})$ and $\support(\beta)\subseteq \support(\hat{\beta})$ respectively. 


\subsection{Asymptotic properties}\label{asymptotics1}

Let $\mathcal{P}$ be the class of laws that obey the intersection submodel $\mathcal{A}\cap\mathcal{B}$; then we are interested in convergence under a sequence of laws $P_n\in \mathcal{P}$. We will allow for $p$ to increase with $n$, and for the values of the parameters $\gamma$ and $\beta$ to depend on $n$, and hence also models $\mathcal{A}$ and $\mathcal{B}$ (although the notation with respect to the models will be suppressed). This is done in order to better gain insight into the finite-sample behavior of the test statistic when $L$ is high-dimensional. Let $\gamma_n$ and $\beta_n$ be the population values of the parameters indexing models $\mathcal{A}$ and $\mathcal{B}$ respectively. Finally, let $\mathbb{P}_{P_n}()$ denote a probability taken with respect to the local data generating process $P_n$.

\begin{theorem}\label{theorem1_drt}
Let us define the active set of variables as $S_\gamma=\support(\gamma_n)$ and $S_\beta=\support(\beta_n)$. Furthermore, let $s_\gamma$ denote the cardinality $|S_\gamma|$ and likewise $s_\beta=|S_\beta|$, and $p$ denote the length of the vectors $\gamma_n$ and $\beta_n$ . Suppose, in addition to Assumptions \ref{moment_drt} and \ref{p_error}  in Appendix A, the following sparsity conditions hold:
\begin{enumerate}[label=(\roman*)]
\item $(s_\gamma+s_\beta )\log(p\lor n)=o(n)$ \label{S_sum}
\item $s_\gamma s_\beta\log^2 (p\lor n)=o(n)$. \label{S_prod}
\end{enumerate}
Then under $H_0$ and the intersection submodel, for all estimators satisfying Assumption \ref{p_error} in Appendix A, we have for any $t\in\mathbb{R}$:
\begin{align}\label{Main res1}
\lim_{n\to \infty}\sup_{P_n\in\mathcal{P}}|\mathbb{P}_{P_n}(T_n\leq t)-\Phi(t)|= 0
\end{align}
\end{theorem}

\begin{remark}
\normalfont The key assumptions required for this result to hold are given in Appendix A. Therein, Assumption \ref{moment_drt} contains mild moment conditions, whereas Assumption \ref{p_error} requires sufficiently fast estimation of $\pi(L;\gamma)$ and $m(L;\beta)$ in the empirical $\ell_2$-norm. Condition \ref{S_sum} requires that both $s_\gamma<<n$ and $s_\beta<<n$; such conditions are quite standard in order to guarantee consistency of the sparse estimators. Condition \ref{S_prod} implies that we can allow for $s_\gamma$ to be large if $s_\beta$ is small, and vice versa. We view this as useful given that in many medical settings, doctors may rely on a limited number of factors when deciding on a patient's treatment. Hence it may even be plausible that the exposure model is `ultra-sparse' e.g. $s_\gamma<<\sqrt{n}$. In contrast, it appears less likely that a model for a clinical outcome (e.g. disease occurrence) can be well approximated by a small number of covariates. These conditions are essentially equivalent to those given in \citet{chernozhukov_double/debiased_2018} and in the Supplementary Appendix of \citet{belloni_inference_2014}, where sample-splitting/cross-fitting is used; indeed, we obtain slightly sharper results than \citet{belloni_inference_2014} (who require that $n^{2/r}s\log(p \lor n)=o(n)$ for some $r>4$ for uniformly valid inference) by focusing on binary exposures.  
\end{remark}



\begin{remark}\label{rem_farrell}
\normalfont 
In \citet{farrell_robust_2015}, uniformly valid inference for the marginal treatment effect is obtained under the stronger assumption that $s^2_\beta=o(n)$ (ignoring log factors). To obtain uniformly valid estimators and tests based on trading-off assumptions on $s_\gamma$ and $s_\beta$, it turns out to be crucial that first order terms like (\ref{problem}) have expectation zero, conditional only on $(L_i)^n_{i=1}$; in many estimation problems this is not possible, because fitting a model for $Y$ requires adjusting for/conditioning on the exposure, so the estimated coefficients depend on $(A_i)^n_{i=1}$. One way to get round this could be to use sample-splitting \citep{chernozhukov_double/debiased_2018}. However, under the null our score test statistic factorizes into a component involving $A$ and $L$, and a component involving $Y$ and $L$. That factorization is much like what one normally achieves via sample-splitting, and hence we can avoid doing this. Indeed, the test statistic also converges to the standard normal under the weaker null hypothesis that $E(Y|A,L)=E(Y|L)$, although this result requires stronger sparsity conditions (namely those of Theorem \ref{theorem2_drt}). Stronger conditions would also be required for the above procedure if the working models are refitted using the union of the selected covariates $ \support(\hat{\beta})\cup\support(\hat{\gamma})$ \citep{belloni_inference_2014}, since the variables selected in the model for $E(Y|L)$ depend additionally on the data $(A_i)^n_{i=1}$ \citep{farrell_robust_2015}.
\end{remark}

\begin{remark}
\normalfont
Although we have focused on sparse estimators, Theorem \ref{theorem1_drt} should hold for more general machine learning methods for estimating $\pi(L;\gamma)$ and $m(L;\beta)$. The key conditions required are that there is consistency in the prediction error, and shrinkage in the product of the $\ell_2$ norms of the errors as $o_p(n^{-1/2})$; see Appendix A and \cite{chernozhukov_double/debiased_2018} for further details.
\end{remark}


\section{Estimation of $\eta$ under model misspecification}\label{nuis_mis}
\subsection{Proposal}

Under the union model $\mathcal{A}\cup\mathcal{B}$, plugging in an arbitrary high-quality sparse estimator $\hat{\eta}$ of $\eta$ into $T_n$ 
will not generally lead to the test statistic converging uniformly to the standard normal. Hence although the score is doubly robust, plugging in $\hat{\eta}$ will not yield \textit{uniformly valid, doubly robust inference}. This can be seen by replicating the Taylor expansion in (\ref{Mest_TE}) for the score $U(\hat{\eta})$; the gradient $\partial U(\eta)/\partial \eta$ is no longer guaranteed to be mean zero and one would generally need to approximate $\sqrt{n}(\hat{\eta}-\eta)$ to assess the variability in the score function under the union model. However, as previously discussed, approximating this term well is generally not possible in the high-dimensional setting. 

We will handle the problematic term (\ref{problem}) by using the gradient $\partial U(\eta)/\partial \eta$ in order to estimate $\eta$, so as to ensure that 
$n^{-1}\sum^n_{i=1}\partial U_i(\eta)/\partial \eta$ is approximately equal to zero at the estimator of the nuisance parameter. This leaves (aside from the remainder) only the score function $U(\eta)$, which we will show is UAN. Specifically, one can estimate $\eta$ by solving the following subgradient estimating equations:
\begin{align}
&0=\frac{1}{n}\sum^n_{i=1}\frac{\partial}{\partial \beta}U_i(\hat{\eta}_{BR})+\lambda_\gamma g (\hat{\gamma}_{BR})\label{br_gamma}\\
&0=\frac{1}{n}\sum^n_{i=1}\frac{\partial}{\partial \gamma}U_i(\hat{\eta}_{BR})+\lambda_\beta g(\hat{\beta}_{BR})\label{br_beta}
\end{align}
Here, $g(a)$ denotes a vector of elements $g(a_j)$, where $g(a_j)=\sign(a_j)$ for $j=1,...,p$ if $a_j\neq 0$ and $g(a_j)\in[-1,1]$ otherwise. The penalty term in (\ref{br_gamma}) corresponds to the subgradient of the $\ell_1$ norm $||\gamma||_1$ with respect to $\gamma$ and likewise for the penalty in (\ref{br_beta}); hence our procedure amounts to $\ell_1$-penalized $m$-estimation. Whilst our procedure requires that the initial working models $\mathcal{A}$ and $\mathcal{B}$ are of the same dimension, using a Lasso penalty will tend to return nuisance parameter estimates with different numbers of non-zero components. We again recommend refitting each working model using the covariates selected via penalization. 
The above procedure extends the bias-reduced doubly robust estimation methodology of \citet{vermeulen_bias-reduced_2015} to incorporate penalization; we thus use $\hat{\gamma}_{BR}$ and $\hat{\beta}_{BR}$ to refer to the resulting estimators of $\gamma$ and $\beta$ respectively.

\subsubsection{Example 1: continuous outcome}\label{cont1}

Returning to the example of Section \ref{test}, we might postulate a linear model for the outcome and a logistic model for the exposure. In this case, the solutions to the equations (\ref{br_gamma}) and (\ref{br_beta}) specify the optima of the following convex optimization problems:
\begin{align}
\label{gamma_est_drt}
\hat{\gamma}_{BR}&=\arg\min_{\gamma}\frac{1}{n}\sum^{n}_{i=1}\log\{1+\exp(\gamma^TL_i)\}-A_i(\gamma^TL_i)+\lambda_\gamma||\gamma||_1\\
\label{beta_est_drt}
\hat{\beta}_{BR}&=\arg\min_{\beta}\frac{1}{2n}\sum^{n}_{i=1}[\expit(\hat{\gamma}_{BR}^TL_i)\{1-\expit(\hat{\gamma}_{BR}^TL_i)\}(Y_i-\beta^TL_i)^2]+\lambda_\beta||\beta||_1.
\end{align}
Hence one can estimate $\gamma$ by fitting a logistic regression model with a Lasso penalty, and then estimate $\beta$ by fitting a linear regression model again with a Lasso penalty and weights constructed using the estimates $\hat{\gamma}_{BR}$. 

\subsubsection{Example 2: binary outcome}\label{bin1}

A more appropriate working model for the conditional mean of binary $Y$ might be $E(Y|A=0,L)=\expit(\beta^TL)$. Hence we now have two minimization problems 
\begin{align*}
\hat{\gamma}_{BR}&=\arg\min_{\gamma}\frac{1}{n}\sum^{n}_{i=1}[\expit(\hat{\beta}_{BR}^TL_i)\{1-\expit(\hat{\beta}_{BR}^TL_i)\}][\log\{1+\exp(\gamma^TL_i)\}-A_i(\gamma^TL_i)]\\&\quad+\lambda_\gamma||\gamma||_1\\
\hat{\beta}_{BR}&=\arg\min_{\beta}\frac{1}{n}\sum^{n}_{i=1}[\expit(\hat{\gamma}_{BR}^TL_i)\{1-\expit(\hat{\gamma}_{BR}^TL_i)\}][\log\{1+\exp(\beta^TL_i)\}-A_i(\beta^TL_i)]\\&\quad+\lambda_\beta||\beta||_1.
\end{align*}
An additional complication is then that solving each set of equations requires initial estimates of the other nuisance parameter. There are then two possible approaches one might take; one is  to estimate $\gamma$ and $\beta$ together by maximizing a joint penalized likelihood. 
Alternatively, one could use the iterative procedure described in Algorithm 1 in Appendix B, which could be easily adapted for other types of outcome. \\

\citet{avagyan_honest_2017} and \citet{tan_model-assisted_2019} make a closely related proposal in the context of estimation the average treatment effect; however, focusing on hypothesis testing and conditional treatment effects enables some simplifications. Our approach for nuisance parameter estimation is based on weighted $\ell_1$-penalized maximum likelihood estimation, so is both easier to implement using existing software and likely to be less computationally demanding in high-dimensional settings. We will also be able to obtain sharper results  (in terms of conditions on sparsity) on the theoretical properties of the test statistic; see Section \ref{asymptotics2}. When the exposure model is linear, our estimator of the unscaled test statistic $n^{-1/2}\sum^n_{i=1}U_i(\hat{\eta}_{BR})$ reduces to the `decorrelated score' approach of \citet{ning_general_2017}. This work thus extends the robustness of their score function to the construction of a test that is UAN under the model $\mathcal{A}\cup\mathcal{B}$. Specifically, by allowing for arbitrary exposure and outcome models and scaling the statistic by a `sandwich estimator' of the variance of $U_i(\hat{\eta}_{BR})$, our test has greater robustness to misspecification than the proposals of \citet{ning_general_2017}, under equivalent assumptions on $s_\gamma$ and $s_\beta$ (see Section \ref{asymptotics2}).

\subsection{Asymptotic properties}\label{asymptotics2}


 
We will now study convergence of $T_n$ under a sequence of laws $P_n\in\mathcal{P}^*$, where $\mathcal{P}^*$ represents a class of laws that obey the union model $\mathcal{A}\cup\mathcal{B}$; hence this class is much larger than that considered in Section \ref{asymptotics1}.

\begin{theorem}\label{theorem2_drt}
Suppose, in addition to Assumptions \ref{moment_drt}, \ref{l_error_w} and \ref{p_error_w} in Appendix A, the following sparsity condition holds:
\begin{enumerate}[label=(\roman*)]\addtocounter{enumi}{2}
\item $(s^2_\gamma+s^2_\beta)\log^2 (p\lor n)=o(n)$. \label{US_sum}
\end{enumerate}
Then under $H_0$ and the union model $\mathcal{A}\cup\mathcal{B}$, using estimators $\hat{\gamma}_{BR}$ and $\hat{\beta}_{BR}$, we have for any $t\in\mathbb{R}$,
\begin{align}\label{Main res2}
\lim_{n\to \infty}\sup_{P_n\in\mathcal{P^*}}|\mathbb{P}_{P_n}(T_n\leq t)-\Phi(t)|= 0.
\end{align}
\end{theorem}

\begin{remark}
\normalfont This theorem states that under the key ultra-sparsity condition \ref{US_sum}, our proposed test is uniformly doubly robust over the parameter space. This condition entails that the number of non-zero coefficients in models $\mathcal{A}$ and $\mathcal{B}$ are small relative to the square root of the overall sample size; this is much stronger than conditions \ref{S_sum} and \ref{S_prod}. Such an assumption is common however in the growing literature on high-dimensional inference \citep{belloni_post-selection_2016,ning_general_2017}, where model misspecification is not generally permitted. Assumptions \ref{l_error_w} and \ref{p_error_w} require sufficiently rapid estimation of the coefficients $\gamma$ and $\beta$ (notably in $\ell_1$-norm) as well as the functions $\pi(L;\gamma)$ and $m(L;\beta)$.
\end{remark}
An advantage of our proposal is that any \textit{a priori} knowledge on the distribution of the exposure can be easily incorporated into the test statistic. Indeed, note that when $E(A|L)$ is known, then following our proposal, one can estimate $\beta$ using standard $\ell_1$-penalized regression without weights; hence the proposal reduces to the one described in Section \ref{nuis_cor}. This is because there is no gradient with respect to the parameters in model $\mathcal{A}$ and the exposure model is guaranteed to be correct. 

\begin{corollary}\label{theorem3_drt}
Suppose that $A$ is randomized, with randomization probability $E(A|L)=\pi(L;\gamma^*)$ for known $\gamma^*$. If $\gamma^*$ is plugged into $T_n$, one can obtain a uniformly valid test using the weaker sparsity condition $s_\beta \log(p\lor n)=o(n)$
regardless of whether model $\mathcal{B}$ is correctly specified.
\end{corollary}

\begin{remark}
\normalfont This corollary of Theorems \ref{theorem1_drt} and \ref{theorem2_drt} states that when we know the randomization probabilities, one can rely merely on this weaker sparsity condition in order to get a valid test, even if model $\mathcal{B}$ is misspecified. 
\end{remark}


When both working models are correctly specified and the outcome regression model is linear, then we have some additional robustness to violations of sparsity, as the following theorem illustrates:

\begin{theorem}\label{theorem4_drt}
When $m(L;\beta)$ is linear in $\beta$, and we restrict ourselves to the class of laws in $\mathcal{P}^*$ that obey the intersection submodel $\mathcal{A}\cap\mathcal{B}$, then (using estimators $\hat{\gamma}_{BR}$ and $\hat{\beta}_{BR}$) the score test statistic converges uniformly as in (\ref{Main res2}) 
under $H_0$, Assumptions \ref{moment_drt}-\ref{bounded_res} and the conditions \ref{S_sum} $(s_\gamma+s_\beta)\log(p\lor n)=o(n)$ and 
\begin{enumerate}[label=(\roman*)]\addtocounter{enumi}{3}
\item $s_\gamma s^*\log^2 (p\lor n)=o(n)$ \label{S_prod_br}
\end{enumerate} where $s^*=s_\gamma \lor s_\beta$.
\end{theorem}

\begin{remark}
\normalfont 
In the context of linear models for $Y$, our proposal is thus `sparsity-adaptive', in the sense that when both models are correct, our proposal is valid under conditions similar to those required for Theorem \ref{theorem1_drt}. As the example in Section \ref{cont1} shows, estimated weights dependent on $(A_i)^n_{i=1}$ are only required in this setting when fitting the outcome model; the proof of Theorem \ref{theorem4_drt} hinges on showing that estimating the weights is of lesser impact than estimating $\beta$ at fixed weights.  For non-linear outcome models, an equivalent result will be difficult to obtain, in light of the fact that fitting the exposure model also requires weights that are dependent on $(Y_i)^n_{i=1}$; however, a general result could be shown using sample-splitting (by estimating the weights in a sample separate to the one used in constructing the test statistic). Nonetheless, this illustrates the trade off between modelling and sparsity conditions; if we wish to obtain inference under the union model then we generally need stronger conditions on $s_\gamma$ or $s_\beta$.
\end{remark}

\begin{remark}
\normalfont The sparsity conditions required for Theorem \ref{theorem4_drt} reduce exactly to those of Theorem \ref{theorem1_drt} (and thus \citet{chernozhukov_double/debiased_2018}) when $s_\gamma \leq s_\beta$. However, the converse does not hold; even if $s_\beta=0$, strong assumptions on $s_\gamma$ are still required. This reflects that estimation of $\beta$ could be harmed by poor quality estimates of the weights. We do not see this asymmetry as a serious disadvantage as in many settings, we would expect that model $\mathcal{A}$ is more likely to be sparse. 
\end{remark}

\section{Simulation Study}\label{simulations}
In this section, we conduct a simulation analysis to compare the performance of the proposed hypothesis test with that of different tests of the causal null hypothesis. In our study, we consider the following tests for the null hypothesis (\ref{cond_ind}): 
\begin{enumerate}
\item A na\"ive post-selection approach where a $t$-test is considered for the null using a linear regression after the standard post-selection of $\beta$ based on $\ell_1$-penalized linear regression. We study the performance of this approach both when the exposure is forced (i.e., the treatment effect $\psi$ is not penalised)  and not forced to be included in the model.

\item Tests based on the `post-double selection' (PDS) and `partialling out' (PO) approaches proposed in the econometrics literature by \citet{belloni_inference_2014}, where $Y$ is regressed on $L$ and $A$ on $L$ also (both using linear models with a Lasso penalty).  PDS then uses a second-stage regression where a linear model for $Y$ is fitted using ordinary least squares, adjusted for $A$ and the union of the covariates selected at the first stage, whereas PO emulates the estimator from \citet{robinson_root-n-consistent_1988} for the partially linear model, also using refitting. Both approaches were implemented in \texttt{R} using the package \texttt{hdm}.


\item The procedure described in Section \ref{nuis_cor}, valid under model $\mathcal{A}\cap\mathcal{B}$, where a score test is considered using standard logistic regression and linear regression after the post-selection of parameters $\gamma$ and $\beta$ based on $\ell_1$-penalized logistic regression and $\ell_1$-penalized linear regression, respectively (hereafter, PMLE-DR).

\item The procedure described in Section \ref{nuis_mis}, valid under model misspecification,  where a score test is considered for the null using standard logistic regression and weighted linear regression after the post-selection of parameters $\gamma$ and $\beta$ based on $\ell_1$-penalized logistic regression (\ref{gamma_est_drt}) and $\ell_1$-penalized weighted linear regression (\ref{beta_est_drt}), respectively (hereafter, BR-DR).
\end{enumerate}

 

Note that all the considered approaches require the selection of penalty parameters. In our simulation study, we use 10-fold cross validation technique to choose the tuning parameters for the na\"ive approach as well as for PMLE-DR and BR-DR. We obtain $\lambda_\gamma$ and $\lambda_\beta$ using \texttt{R} package \texttt{glmnet} through the argument \texttt{lambda.min}; this selects the value which minimises the mean cross-validated error. In the \texttt{hdm} package, PDS and PO and implemented using pre-specified values for the penalty parameters \citep{chernozhukov_hdm:_2016};
in order to study the impact of using different penalties, we also performed PDS as described in \citet{belloni_inference_2014} using cross-validation instead of the pre-specified values. 

In the simulation analysis, we generate $n$ mutually independent vectors $\{(Y_i,A_i,L^T_i)^T\}$, $i=1,...,n$. 
Here, ${L}_i=(L_{i,1}, ..., L_{i,p})$ is a mean zero multivariate normal covariate with covariance matrix $\Sigma$;  either $\Sigma=\textnormal{I}_{p\times p}$ (uncorrelated covariates) or $\Sigma=[\sigma_{i,j}]_{1\leq i,j\leq p}$, where $\sigma_{i,j}=2^{-1-|i-j|}$ (correlated covariates). For simplicity we consider a binary exposure model and linear outcome model. We let for each $i=1,...,n$, the dichotomous exposure $A_i$ take on values 0 or 1 with $P(A_i=1|L_i)\equiv \pi({L}_i)$, the outcome $Y_i$ be normally distributed with mean $ m({L}_i)$ and unit variance, conditional on $L_i$  and $A_i$. Further, the simulated data are analysed using the following parametric working models: $\pi(L;\gamma)=\textnormal{expit}(\gamma_0+\displaystyle\sum_{i=1}^p\gamma_iL_i)$ and $m(L,\beta)=\beta_0+\displaystyle\sum_{i=1}^p\beta_iL_i$ , where $\beta_0=1$, $\gamma_0=2$. The nuisance parameters $\beta=(\beta_1, ..., \beta_p)\in \mathbb{R}^{p}$ and $\gamma=(\gamma_1, ..., \gamma_p)\in \mathbb{R}^{p}$ are defined as\\
$b=\left(2\frac{\log(20)}{n^{1/2}}, 2\frac{\log(19)}{ n^{1/2}}, ...,2 \frac{\log(2)}{n^{1/2}},0_{20},...,0_{81}, 10 \frac{\log(2)}{n^{1/2}},...,10\frac{\log(20)}{n^{1/2}},0_{101},...,0_{p}\right)$, $\beta=2 \cdot b \cdot \left({\sum_{i=1}^p b^2}\right)^{-1/2}$, $g=\left(40\frac{\log(20)}{n^{1/2}}, 40 \frac{\log(19)}{n^{1/2}}, ...,40 \frac{\log(2)}{n^{1/2}},0_{20},...,0_{p}\right)$,
$\gamma={3} \cdot g \cdot \left({\sum_{i=1}^p g^2}\right)^{-1/2}$
where the subscripts indicate the index (i.e., position) of $0$ in the vector. The considered settings for nuisance parameters are challenging in the sense that there are confounders that are strongly predictive  of the exposure and weakly predictive of the outcome. Moreover, there are covariates which are moderately predictive of the outcome but are not associated with the exposure. In order to evaluate the impact of model misspecification, we next generate data with the following outcome model: $m(L, \beta)=1+\beta^T\left(|L_{.,[1:3]}|; L_{.,[4:p]}\right)$. Finally, for the data generating mechanism described above, we perform 1,000 Monte Carlo runs for $n=200$ and $p=200$, $n=500$ and $p=500$, $n=200$ and $p=100$, and $n=200$ and $p=250$.

Tables \ref{Table1} and \ref{Table2} show the Type I errors based on 1,000 replications. The simulation results show that the PMLE-DR and BR-DR approaches have rejection rates close to the nominal level of $5\%$, so long as the outcome model is correctly specified. On the other hand, we observe that even when both models are correctly specified, the na\"ive approaches provide high rejection rates. Moreover, these rates do not diminish with larger sample size. This poor performance is well aligned with the theory of \citet{leeb_model_2005}. We also observe that the rejection rates of PDS and PO are relatively high. Part of this poor performance appeared to be due to the particular data-driven procedure for selecting the penalty parameters, which led to an insufficient number of covariates being selected. However, even when using cross validation for PDS, there was still a discrepancy between the methods. When the covariates were correlated and the outcome model was incorrect, the PMLE-DR test was mildly anti-conservative relative to the BR-DR test; the fact that PDS (which is not generally doubly robust) performed relatively well in this setting indicates that the type of misspecification considered may not be particularly damaging.  In Appendix C, we also consider additional settings under the modified sparsity in the propensity score model, as well as heteroscedasticity; similar results to those in \ref{Table1} and \ref{Table2} are seen across settings.




\begin{table}[h]
\centering
   \caption{ Type I errors at the 5\% significance level based on 1,000 replications: $\Sigma=\textnormal{I}_{p\times p}$.}
   \label{Table1}
   \begin{tabular}{l  r r r r r r r}
\\
    \textit{Correct models} &&\\

 \hline
   			   & $n=200$   &   &  $n=500$  & & $n=200$ & & $n=200$\\  
    Methods    &  $p=200$ &    &  $p=500$ & & $p=100$ & &  $p=250$ \\ \hline
Standard na\"ive (forced) & 0.548 & & 0.809 & & 0.291  && 0.600 \\  
Standard na\"ive (not forced) & 0.275 & & 0.555 & & 0.168 && 0.316\\ 
PDS (pre-specified) & 0.517 & & 0.761 & & 0.536 && 0.498\\
PO (pre-specified) & 0.508 & & 0.748 &&  0.519 && 0.482\\
PDS (CV) & {0.074} && {0.074}  & & 0.068 &&0.072\\
PMLE-DR & 0.055 & & 0.054 & & 0.070 && 0.075\\ 

BR-DR  & 0.053 & & 0.069  & & 0.078 && 0.073 \\\hline\\

\textit{Incorrect outcome model} &&\\
\hline
 & $n=200$  & & $n=500$  && $n=200$ && $n=200$ \\ 
Methods  &  $p=200$ &  &  $p=500$  && $p=100$ && $p=250$\\ \hline
Standard na\"ive (forced) & 0.368 && 0.586 && 0.210 && 0.425\\  
Standard na\"ive (not forced) & 0.175 && 0.313&& 0.115  && 0.197\\
PDS (pre-specified) & 0.319 && 0.634&& 0.345 &&0.327 \\
PO (pre-specified)  & 0.317 & &0.615&& 0.345 && 0.309\\
PDS (CV)  & {0.073} && {0.072}&& 0.060 && 0.070\\

PMLE-DR & 0.056 && 0.053&& 0.070 && 0.059\\
BR-DR  & 0.046 && 0.059&& 0.081 && 0.050 \\\hline 

\end{tabular}
\end{table}

\begin{table}[h]
\centering
   \caption{ Type I errors at the 5\% significance level based on 1,000 replications: $\Sigma=[\sigma_{i,j}]_{1\leq i,j\leq p}$.}
   \label{Table2}
   \begin{tabular}{l  r r r r r r r}
\\
   		\textit{Correct models} &&\\
   		
   		\hline
   		& $n=200$   &   &  $n=500$  & & $n=200$ & & $n=200$\\  
   		Methods    &  $p=200$ &    &  $p=500$ & & $p=100$ & &  $p=250$ \\ \hline
   		Standard na\"ive (forced) & 0.448 & & 0.387 & & 0.190  && 0.522 \\  
   		Standard na\"ive (not forced) & 0.178 & & 0.164 & & 0.100 && 0.193\\ 
   		PDS (pre-specified) & 0.143 & & 0.064 & & 0.103 && 0.131\\
   		PO (pre-specified) & 0.115 & & 0.064 &&  0.085 && 0.108\\
   		PDS (CV) & 0.066 & & 0.055 & & 0.059 &&0.075\\
   		PMLE-DR & 0.057 & & 0.063 & & 0.063 && 0.047 \\ 
   		BR-DR  &0.043 & & 0.049 & & 0.051 && 0.046\\\hline\\

   		\textit{Incorrect outcome model} &&\\
   		\hline
   		& $n=200$  & & $n=500$  && $n=200$ && $n=200$ \\ 
   		Methods  &  $p=200$ &  &  $p=500$  && $p=100$ && $p=250$\\ \hline
   		Standard na\"ive (forced) & 0.312 && 0.277 && 0.144 && 0.363\\  
   		Standard na\"ive (not forced) & 0.139 && 0.113 && 0.070  && 0.140\\
   		PDS (pre-specified) & 0.094 && 0.057 && 0.079 &&0.091 \\
   		PO (pre-specified)  & 0.077 & & 0.050 && 0.060 && 0.077\\
   		PDS (CV)  & 0.071 && 0.061 && 0.057 && 0.073\\
   		
   		PMLE-DR & 0.063 && 0.057 && 0.060 && 0.056 \\
   		BR-DR  & 0.030 && 0.051 && 0.044 && 0.041 \\\hline 

\end{tabular}
\end{table}

\section{Data analysis}\label{data_analysis}

Glycemic control in critically ill patients is still the subject of controversy, in terms of the optimal limits in which glucose levels are best kept. In the Leuven II randomised trial \citep{van_den_berghe_intensive_2001}, strict glycemic control (with the maintenance of glycemia between 80 and 110 milligram per deciliter (mg/dl)) resulted in reduced mortality. Later multi-center studies could not replicate these findings, including the NICE-SUGAR trial \citep{finfer_intensive_2009}. Current guidelines usually recommend glycemic control between 140 and 180mg/dl. In the Ghent University Intensive Care Unit (UZ Ghent ICU) a glycemic protocol is used, targeting values between 80 and 150 mg/dl. In practice, glycemia in patients often falls outside of this range, partly due to a lack of compliance in following the protocol. We sought to investigate the relationship between glycemic control and 30-day mortality, using routinely collected data from the UZ Ghent ICU on a large representative cohort of intensive care patients. Specifically, we aimed to test the null hypothesis of no effect of a change in glycemia level (from $<$110  to $\geq$110 mg/dl, and then from $\leq$150 to $>$150 mg/dl) at any day of follow-up on death within 30 days from ICU entry. We restricted the analysis to patients that were alive in the ICU for at least 48 hours, thus removing patients who died immediately upon arrival in the ICU. 

Data were obtained from the electronic patient data management system of the UZ Ghent ICU. The potential confounders were split up into variables assessed at admission into the intensive care unit and variables where data were collected over time. For covariates that were measured repeatedly, we took the mean of the measurements taken within the previous 48 hours to the considered day of follow up in the ICU for continuous covariates, and the maximum value for categorical covariates. Measurements on glycemia were usually recorded multiple times per day, so in order to create the exposure, we took the mean of the measurements from within the first 6 hours of the day. For this illustration, any patients with missing data on the exposure, outcome or confounders were removed from the dataset. In order to perform our test, at each day we assumed (in individuals still alive) a logistic regression model for the probability of glycemia level $\geq$110 mg/dl (or $>$150 mg/dl) as well as a logistic model for death within 30 days of entering hospital. To avoid the issues associated with time-varying confounding described e.g. in \citet{robins_causal_1997}, in each regression model we adjusted only for covariate data, as well as previous exposures, collected prior to the glycemia measurements on a given day. We then used an amended version of the test for binary outcomes described in Section \ref{bin1} (implemented using Algorithm 1 described in Appendix C), allowing for potential misspecification in either the exposure or outcome model. Given that the data consisted of multiple observations per individual, then letting $t$ denote a particular day,  we replaced $U_i(\hat{\eta})$ with $\sum_t U_{it}(\hat{\eta})$ in the statistic $T_n$ (for $t=3,...,30$). In our modeling, we included all confounders selected by clinical experts, as well as quadratic terms of continuous variables and all two-way interactions between main effects. 

We obtained data on 12,105 patients entering the intensive care unit; after restricting to patients still alive at day 3, 10,885 individuals remained. Further removing patients entering prior to 2013 left 4,682 individuals, with a final dataset of 4,120 after removing those with missing data. Given that patients were assessed on multiple days, there were 24,863 observations in the dataset; the median number of contributed observations was 3 (the minimum was 1 and the maximum 28). In this final cohort, 768 (18.6\%) of individuals died in hospital within 30 days of entering the ICU. Considering the mean glycemia values for patients within the first 6 hours of day 3, the average of these values among all patients was 131.6 (minimum: 45, maximum: 492). 927 (23.3\%) patients had mean glycemia at day 3 $<$110 mg/dl, 2,208 (55.5\%) had a level $\geq$110mg/dl and $\leq$150mg/dl and 841 (21.2\%) had a level $>$150 mg/dl. After generating interactions, there were 148 covariates to adjust for in the analysis. Looking at a change at each day from $<$110 to $\geq$110 mg/dl, the test statistic $T_n$ was -1.42 with a $p$-value of 0.156 whereas changing from $\leq$150 to $>$150 mg/dl gave a test statistic of 6.98 ($p<$0.001). Hence, at the 5\% level, we saw evidence of a difference in 30 day mortality based on a change from moderate ($\leq$150mg/dl) to high ($>$150 mg/dl) glycemia levels on a given day. On the other hand, in comparing those with low ($<$110mg/dl) vs. higher ($\geq$110 mg/dl) glycemia levels, we did not observe a statistically significant difference at the 5\% level.

\section{Discussion}\label{discussion}

We have proposed a general framework for constructing uniformly valid tests of GLM parameters in high-dimensional settings. We hope to have clarified why locally doubly robust methods (in this case, doubly robust under the null) have a privileged position in the literature \citep{farrell_robust_2015}; if all working models are correct, one can obtain a uniformly valid test by plugging in any sufficiently fast-converging sparse estimator of the nuisance parameters. If one of the working models is misspecified, then one can still obtain uniformly valid inference, so long as a specific estimation procedure for the nuisance parameters is used. We have also indicated why score tests might be preferable in high-dimensional settings, since then the outcome model can be fit under the null hypothesis, enabling one to weaken conditions on sparsity. 

In future work, we will extend our procedures to the estimation of regression parameters and the construction of confidence intervals. Consider the model $\mathcal{M}$ defined by the restriction \[g\{E(Y|A=a,L=l)\}-g\{E(Y|A=0,L=l)\}=\psi a\] where $g(\cdot)$ is a known link function. The score $U(\eta)$ implies an estimator of $\psi$, the conditional causal effect of $A$ on $Y$. Let $H(\psi)=Y-\psi A$ when $g()$ is the identity link and $H(\psi)=Y\exp(-\psi A)$ when $g()$ is the log link; then estimation of $\psi$ can be based on the function
\begin{align}\label{dr_ef}
U(\psi,\eta)=\{A-\pi(L;\gamma)\}\{H(\psi)-m(L;\beta)\}
\end{align}
\citep{robins_estimating_1992}. An estimator of $\psi$ based on (\ref{dr_ef}) is consistent under model $\mathcal{M}\cap(\mathcal{A}\cup\mathcal{B})$. The goal of constructing uniformly valid confidence intervals could require a revision of the conditions given in Sections \ref{asymptotics1} and \ref{asymptotics2}, since we are no longer working under the null. It also remains an open question for which settings doubly robust estimators can be constructed. For example, there currently exists no doubly robust estimator for the Cox proportional hazards model or probit models. In practice it may be more feasible to construct estimators and confidence intervals that are locally doubly robust e.g. under the null, and in this context enjoy the properties of the tests described in this paper.

When $\psi$ is multivariate, equations (\ref{br_gamma}) and (\ref{br_beta}) deliver more estimating equations than there are unknown nuisance parameters. To ensure that standard errors are valid, one would also need to ensure that the estimating functions of each component of $\psi$ are orthogonal to those of the remaining components. Such a development would not only be advantageous in terms of testing for and estimating interaction terms, but also for obtaining uniformly valid inference in high-dimensional data with mediators and/or time dependent confounders. Indeed, the estimators described above are special cases of \textit{g-estimators}  \citep{robins_estimating_1992}, developed for fitting \textit{structural nested models} in complex longitudinal studies. Because it turns out to be essentially impossible to correctly specify sequential regression models for an outcome, it is unlikely that existing proposals for high-dimensional inference can be adapted to test the hypothesis of no causal effect of any treatment \textit{regime} on $Y$ a.k.a the \textit{g-null hypothesis} \citep{robins_causal_1997}. In contrast, although we perform selection on both the outcome and exposure models (in order for the relevant gradients to be set to zero), in the proposal of Section \ref{nuis_mis} only the latter needs to be correctly specified in order to obtain a valid test of the g-null. 

\section*{Acknowledgements}
Oliver Dukes is supported by the Strategic Basic Research PhD grant 1S05916N from the Research Foundation - Flanders (FWO).  Vahe Avagyan and Stijn Vansteelandt are supported by the FWO research project G016116N and Special Research Fund (BOF) research project BOF.244.2017.0004.01. The authors are grateful to Prof. Johan Decruyenaere and Prof. Kirsten Colpaert for permission to use the Ghent University Intensive Care Unit dataset, and to Bram Gadeyne and Martijn Busselen for assistance in managing the data.

\newpage

\appendix
\numberwithin{equation}{section}

\section{Appendix A}\label{appA_drt}

Define $\mathbb{N}_{p}=\{1,2,...,p\}$. We use $\mathbb{E}_{P_n}[]$ for taking expectation w.r.t. the local data generating process (DGP), whereas $\mathbb{E}_n[]$ refers to sample expectations. Similarly, $\mathbb{P}_{P_n}[]$ and  $\var_{P_n}[]$ denote probabilities and variances taken w.r.t. the local DGP respectively. 
We define the $\ell_\infty$ norm of any matrix $A$ as $||A||_{\infty}=\displaystyle\max_{i,j}|A_{ij}|$. 
For a vector $\delta\in\mathbb{R}^p$ and indices $T\subset\{1,...,p\}$, let $\delta_T$ denote the vector where $\delta_{Tj}=\delta_j$ if $j\in T$ and $\delta_{Tj}=0$ otherwise. Also, $T^C=\{1,...,p\}\backslash T$. For sequences $a_n$ and $b_n$, we use $a_n\lesssim b_n$ to denote that $a_n\leq Cb_n$ for some constant $C$ ($a_n\gtrsim b_n$ is similarly defined). We view the gradients $\partial U_i(\hat{\eta})/\partial \eta$ and $\partial U_i(\eta_n)/\partial \eta$ as row vectors.

\subsection{Proof of Theorem \ref{theorem1_drt}}

In order to prove Theorem 1, we will rely on the following assumptions.

\begin{assumption}\label{moment_drt}
(Data generating process) There exist constants $C_1,  C_2, C_3<\infty$,  $c_4, c_5>0$ and $4<r<\infty$ such that:
\begin{enumerate}[label=(\roman*)]
\item $\mathbb{E}_{P_n}\{|Y-m(L;\beta_n)|^{4}|L\}\leq C_1$ w.p. 1. \label{m1_drt}
\item $\mathbb{E}_{P_n}\{|Y-m(L;\beta_n)|^r\}\leq C_2$. \label {m2_drt}
\item $\max_{i\leq n}\norm{L_i}_\infty \leq C_3<\infty$ w.p. 1. \label{m3_drt}
\item $c_4\leq\mathbb{E}_{P_n}[\{A-\pi(L;\gamma_n)\}^{2}|L]$ and $c_5\leq\mathbb{E}_{P_n}[\{Y-m(L;\beta_n)\}^{2}|L]$ w.p. 1. \label{m4_drt}
\end{enumerate}
\end{assumption}

\begin{remark}
\normalfont Assumption \ref{moment_drt}\ref{m1_drt} allows one to bound the conditional variance of $Y-m(L;\beta_n)$ given $L$ and also implies a bound on the variance given $A$ and $L$. Assumption \ref{moment_drt}\ref {m2_drt} places a bound on the higher order moments of $Y-m(L;\beta_n)$, and is required to show uniform consistency of the variance estimator of $U(\hat{\eta})$ and uniform asymptotic normality of the test statistic. We note that Assumptions \ref{moment_drt}\ref{m1_drt}-\ref {m2_drt} allow for non-Gaussianity and heteroscedasticity with respect to the error term $Y-m(L;\beta_n)$. Assumption \ref{moment_drt}\ref{m3_drt} requires $L$ to be restricted to a bounded set, which is an assumption commonly made in the literature (sometimes as a primitive condition for proving consistency of estimators) \citep{van_de_geer_asymptotically_2014,farrell_robust_2015,ning_general_2017}. 
Assumption \ref{moment_drt}\ref{m4_drt} places additional bounds on the conditional variance, and implies a type of `positivity' condition such that there must be some variation in $A$ at different levels of $L$.
\end{remark}

\begin{assumption}\label{p_error}
(Rates for prediction error with unweighted estimators) 
\begin{enumerate}[label=(\roman*)]
\item $\mathbb{E}_{n}[\{\pi(L_i;\gamma_n)-\pi(L_i;\hat{\gamma})\}^2]=O_{P_n}(s_\gamma\log(p \lor n)/n).$ \label{gamma_p_error}
\item $\mathbb{E}_{n}[\{m(L_i;\beta_n)-m(L_i;\hat{\beta})\}^2]=O_{P_n}(s_\beta\log(p \lor n)/n).$ \label{beta_p_error}
\end{enumerate}
\end{assumption}
\begin{remark}
\normalfont 
Results \ref{p_error}\ref{gamma_p_error}-\ref{p_error}\ref{beta_p_error} follow from the results of \citet{belloni_least_2013}, \citet{belloni_inference_2014}, \citet{farrell_robust_2015} and \citet{belloni_post-selection_2016} on Lasso and Post-Lasso-based estimators. Rates on quantities like $\mathbb{E}_{n}[\{\gamma^T_nL_i-\hat{\gamma}^TL_i\}^2]$ also follow from those papers.
\end{remark}
\begin{proof}
The proof will proceed in four steps. In the first step, we show that 
\begin{align}\label{Res_1}
\frac{1}{\sqrt{n}}\sum^n_{i=1}U_i(\hat{\eta})=\frac{1}{\sqrt{n}}\sum^n_{i=1}U_i(\eta_n)+o_{P_n}(1)
\end{align}
in the second, that
\begin{align}\label{Res_2}
\frac{\mathbb{E}_{n}\{U_i(\eta_n)\}}{\sqrt{\frac{1}{n}\mathbb{E}_{P_n}\{U_i(\eta_n)^2\}}}\overset{d}{\to}\mathcal{N}(0,1)
\end{align}
in the third, that
\begin{align}\label{Res_3}
\mathbb{E}_{n}[U_i(\hat{\eta})^2-\mathbb{E}_{n}\{U_i(\hat{\eta})\}^2]^{-1}=\mathbb{E}_{P_n}\{U_i(\eta_n)^2\}^{-1}+o_{P_n}(1)
\end{align}
Finally, we will use these results to show result (\ref{Main res1}) in the main paper. 

\subsubsection*{Step 1.}\quad

Consider the sample mean of  $U_i(\hat{\eta})$:
\begin{align*}
\mathbb{E}_n\{U_i(\hat{\eta})\}=\mathbb{E}_n\{U_i(\eta_n)+U_i(\hat{\eta})-U_i(\eta_n)\}
\end{align*}
After some algebra, we have
\begin{align*}
\sqrt{n}\mathbb{E}_n\{U_i(\hat{\eta})\}=\sqrt{n}\mathbb{E}_n\{U_i(\eta_n)\}+R_1+R_2+R_3
\end{align*}
where 
\begin{align*}
&R_1=\frac{1}{\sqrt{n}}\sum^n_{i=1}\{A_i-\pi(L_i;\gamma_n)\}\{m(L_i;\beta_n)-m(L_i;\hat{\beta})\},\\
&R_2=\frac{1}{\sqrt{n}}\sum^n_{i=1}\{Y_i-m(L_i;\beta_n)\}\{\pi(L_i;\gamma_n)-\pi(L_i;\hat{\gamma})\}\\
&R_3=\frac{1}{\sqrt{n}}\sum^n_{i=1}\{m(L_i;\hat{\beta})-m(L_i;\beta_n)\}\{\pi(L_i;\hat{\gamma})-\pi(L_i;\gamma_n)\}
\end{align*}
We aim to show that $R_1$, $R_2$ and $R_3$ are all $o_{P_n}(1)$ under model $\mathcal{A}\cap\mathcal{B}$.

For $R_1$, under the null as defined in (\ref{cond_ind}),
\begin{align*}
&\mathbb{E}_{P_n}\{R_{1}|(Y_i,L_i)^n_{i=1}\}\\&=\frac{1}{\sqrt{n}}\sum^n_{i=1}[\mathbb{E}_{P_n}\{A_i|(Y_i,L_i)^n_{i=1}\}-\pi(L_i;\gamma_n)]\{m(L_i;\beta_n)-m(L_i;\hat{\beta})\}\\
&=0
\end{align*}
and 
\begin{align*}
&\mathbb{E}_{P_n}\{R^{2}_{1}|(Y_i,L_i)^n_{i=1}\}\\&=\mathbb{E}_{n}\left(\mathbb{E}_{P_n}[\{A_i-\pi(L_i;\gamma_n)\}^2|(Y_i,L_i)^n_{i=1}]\{m(L_i;\beta_n)-m(L_i;\hat{\beta})\}^2\right)
\\&\leq C\mathbb{E}_{n}[\{m(L_i;\beta_n)-m(L_i;\hat{\beta})\}^2]
\end{align*}
where $C$ is a constant. Furthermore, invoking Assumption \ref{p_error}\ref{beta_p_error} and sparsity condition \ref{S_sum}, we have 
\[C\mathbb{E}_{n}[\{m(L_i;\beta_n)-m(L_i;\hat{\beta})\}^2]=o_{P_n}(1)\]
and  $\mathbb{E}_{P_n}[R^2_1]=o(1)$. 
Hence one can then apply Chebyshev's inequality to show that $|R_1|=o_{P_n}(1)$.

Similarly, for $R_2$,
\begin{align*}
&\mathbb{E}_{P_n}[R^2_2|(A_i,L_i)^n_{i=1}]\\&=\mathbb{E}_{n}\bigg[\mathbb{E}_{P_n}[\{Y_i-m(L_i;\beta_n)\}^2|(A_i,L_i)^n_{i=1}]\{\pi(L_i;\gamma_n)-\pi(L_i;\hat{\gamma})\}^2\bigg]\\&\leq C\mathbb{E}_{n}[\{\pi(L_i;\gamma_n)-\pi(L_i;\hat{\gamma})\}^2],
\end{align*}
where $C$ is a constant. This inequality follows from Assumption \ref{moment_drt}\ref{m1_drt}. Invoking Assumption \ref{p_error}\ref{gamma_p_error} and sparsity condition \ref{S_sum}, we have 
\[C\mathbb{E}_{n}[\{\pi(L_i;\gamma_n)-\pi(L_i;\hat{\gamma})\}^2]=o_{P_n}(1)\]
so  $\mathbb{E}_{P_n}[R^2_2]=o(1)$ and using Chebyshev's inequality, $|R_2|=o_{P_n}(1)$. 

Finally, considering $R_3$, by H\"older's inequality
\begin{align*}
&\bigg|\frac{1}{\sqrt{n}}\sum^n_{i=1}\{m(L_i;\hat{\beta})-m(L_i;\beta_n)\}\{\pi(L_i;\hat{\gamma})-\pi(L_i;\gamma_n)\}\bigg|\\
&\leq \sqrt{n}\mathbb{E}_n[\{m(L_i;\hat{\beta})-m(L_i;\beta_n)\}^2]^{1/2}\mathbb{E}_n[\{\pi(L_i;\hat{\gamma})-\pi(L_i;\gamma_n)\}^2]^{1/2}
\end{align*}
Then given the joint sparsity condition \ref{S_prod} on $s_\gamma$ and $s_\beta$, and Assumptions \ref{p_error}\ref{gamma_p_error} and \ref{p_error}\ref{beta_p_error}, it follows that
\[\sqrt{n}\mathbb{E}_n[\{m(L_i;\hat{\beta})-m(L_i;\beta_n)\}^2]^{1/2}\mathbb{E}_n[\{\pi(L_i;\hat{\gamma})-\pi(L_i;\gamma_n)\}^2]^{1/2}=o_{P_n}(1)\]
Therefore $|R_3|=o_{P_n}(1)$ and we have result (\ref{Res_1}).

\subsubsection*{Step 2.}\quad

Under the null, we have that 
\begin{align*}
\var_{P_n}\{U_i(\eta_n)\}&=\mathbb{E}_{P_n}\{U_i(\eta_n)^2\}\\&=\mathbb{E}_{P_n}[\{A_i-\pi(L_i;\gamma_n)\}^2\{Y_i-m(L_i;\beta_n)\}^2]
\end{align*}
and by Assumptions \ref{moment_drt}\ref{m2_drt} and \ref{moment_drt}\ref{m4_drt}, $\mathbb{E}_{P_n}\{U_i(\eta_n)^2\}$ is bounded away from zero (necessary for the inversion) and above uniformly in $n$. 

For some $\epsilon>0$, such that $4+2\epsilon\leq r$
\begin{align*}
&\mathbb{E}_{P_n}\{|U_i(\eta_n)|^{2+\epsilon}\}\\
&\leq \mathbb{E}_{P_n}\{|A_i-\pi(L_i;\gamma_n)|^{4+2\epsilon}\}^{1/2}\mathbb{E}_{P_n}\{|Y_i-m(L_i;\beta_n)|^{4+2\epsilon}\}^{1/2}\\
&\leq C
\end{align*}
where $C$ is a constant, by Assumption \ref{moment_drt}\ref {m2_drt}. This verifies the Lyapunov condition, such that using this result (and the fact that $\mathbb{E}_{P_n}\{U_i(\eta_n)^2\}$ is finite) one can then invoke the Lyapunov central limit theorem for triangular arrays to get result (\ref{Res_2}). We rely on array asymptotics here in order to allow for the data-generating process to change with $n$.


\subsubsection*{Step 3.}\quad


Since $\mathbb{E}_{P_n}\{U_i(\eta_n)^2\}$ is bounded away from zero uniformly in $n$ and $\mathbb{E}_{P_n}\{U_i(\eta_n)\}=0$, given the previous steps it suffices to show that $\mathbb{E}_{n}\{U_i(\hat{\eta})^2\}=\mathbb{E}_{P_n}\{U_i(\eta_n)^2\}+o_{P_n}(1)$. We will first obtain the result
\begin{align}\label{var_res_init}
\mathbb{E}_{n}\{U_i(\eta_n)^2\}=\mathbb{E}_{P_n}\{U_i(\eta_n)^2\}+o_{P_n}(1)
\end{align}

We have
\begin{align*}
&\mathbb{P}_{P_n}\left[|\mathbb{E}_n\{U_i(\eta_n)^2\}-\mathbb{E}_{P_n}\{U_i(\eta_n)^2\}|^2>\epsilon\right]\\
&\leq \frac{1}{\epsilon^2}\mathbb{E}_{P_n}\bigg[\bigg|\frac{1}{n}\sum^n_{i=1}U_i(\eta_n)^2-\mathbb{E}_{P_n}\{U_i(\eta_n)^2\}\bigg|^2\bigg]\\
&\leq \frac{1}{\epsilon^2n^2}\bigg(2-\frac{1}{n}\bigg)\sum^n_{i=1}\mathbb{E}_{P_n}\bigg[\big|U_i(\eta_n)^2-\mathbb{E}_{P_n}\{U_i(\eta_n)^2\}\big|^2\bigg]
\end{align*}
where we first apply Chebyshev's inequality. The second uses the Von Bahr-Esseen inequality: let $q\in[1,2]$, then for independent mean-zero variables $X_1,...,X_n$, we have
\[E\bigg(\bigg|\sum_{i=1}^nX_i\bigg|^q\bigg)\leq\bigg(2-\frac{1}{n}\bigg)\sum_{i=1}^nE(|X_i|^q)\]
\citep{von_bahr_inequalities_1965}.

Since
\begin{align*}
&\mathbb{E}_{P_n}\left([U_i(\eta_n)^2-\mathbb{E}_{P_n}\{U_i(\eta_n)^2\}]^2\right)\\
&=\var_{P_n}\{U_i(\eta_n)^2\}\\
&=\var_{P_n}\big(\{A_i-\pi(L_i;\gamma_n)\}^2\mathbb{E}_{P_n}[\{Y_i-m(L_i;\beta_n)\}^2|A_i,L_i]\big)\\
&\quad+\mathbb{E}_{P_n}\big(\{A_i-\pi(L_i;\gamma_n)\}^4\var_{P_n}[\{Y_i-m(L_i;\beta_n)\}^2|A_i,L_i]\big)
\end{align*}
then firstly
\begin{align*}
&\var_{P_n}\big(\{A_i-\pi(L_i;\gamma_n)\}^2\mathbb{E}_{P_n}[\{Y_i-m(L_i;\beta_n)\}^2|A_i,L_i]\big)\\
&\leq C \mathbb{E}_{P_n}[\{A_i-\pi(L_i;\gamma_n)\}^2]^2=O(1)
\end{align*}
where $C$ is a constant, using Assumption \ref{moment_drt}\ref{m1_drt}. Secondly,
\begin{align*}
&\mathbb{E}_{P_n}\big(\{A_i-\pi(L_i;\gamma_n)\}^4\var_{P_n}[\{Y_i-m(L_i;\beta_n)\}^2|A_i,L_i]\big)\\
&\leq C\mathbb{E}_{P_n}[\{A_i-\pi(L_i;\gamma_n)\}^4]=O(1)
\end{align*}
where $C$ is again a constant, invoking Assumptions \ref{moment_drt}\ref{m1_drt} and \ref{moment_drt}\ref {m2_drt}. Result (\ref{var_res_init}) then follows. 

It remains to show that 
\begin{align}\label{temp_res}
\mathbb{E}_{n}\{U_i(\hat{\eta})^2\}=\mathbb{E}_{n}\{U_i(\eta_n)^2\}+o_{P_n}(1)
\end{align}
By adding and subtracting $\mathbb{E}_{n}[\{A_i-\pi(L_i;\hat{\gamma})\}^2\{Y_i-m(L_i;\beta_n)\}^2]$ and applying the triangle inequality, then
\begin{align*}
&|\mathbb{E}_{n}[\{A_i-\pi(L_i;\hat{\gamma})\}^2\{Y_i-m(L_i;\hat{\beta})\}^2-\{A_i-\pi(L_i;\gamma_n)\}^2\{Y_i-m(L_i;\beta_n)\}^2]|\\
&\leq\big|\mathbb{E}_{n}\big([\{A_i-\pi(L_i;\hat{\gamma})\}^2-\{A_i-\pi(L_i;\gamma_n)\}^2]\{Y_i-m(L_i;\beta_n)\}^2\big)\big|\\
&\quad+\big|\mathbb{E}_{n}\big([\{Y_i-m(L_i;\hat{\beta})\}^2-\{Y_i-m(L_i;\beta_n)\}^2]\{A_i-\pi(L_i;\hat{\gamma})\}^2\big)\big|\\
&=|R_4|+|R_5|
\end{align*}
Looking first at $|R_5|$, after some algebra we have
\begin{align*}
&\big|\mathbb{E}_{n}\big([\{Y_i-m(L_i;\hat{\beta})\}^2-\{Y_i-m(L_i;\beta_n)\}^2]\{A_i-\pi(L_i;\hat{\gamma})\}^2\big)\big|\nonumber\\
&\leq |\mathbb{E}_{n}[\{m(L_i;\hat{\beta})-m(L_i;\beta_n)\}^2\{A_i-\pi(L_i;\hat{\gamma})\}^2 ]|\\
&\quad+|2\mathbb{E}_{n}[\{Y_i-m(L_i;\beta_n)\}\{m(L_i;\hat{\beta})-m(L_i;\beta_n)\}\{A_i-\pi(L_i;\hat{\gamma})\}^2]|\\
&=|R_{5a}|+|R_{5b}|.
\end{align*}
Then,
\begin{align*}
|R_{5a}|\leq \max_{i\leq n}\{A_i-\pi(L_i;\hat{\gamma})\}^2\mathbb{E}_{n}[\{m(L_i;\hat{\beta})-m(L_i;\beta_n)\}^2]=o_{P_n}(1)
\end{align*}
following Assumption \ref{p_error}\ref{beta_p_error}, the fact that $A$ is binary and sparsity condition \ref{S_prod}. Furthermore,
\begin{align*}
|R_{5b}|&\leq 2 \max_{i\leq n}\{A_i-\pi(L_i;\hat{\gamma})\}^2\mathbb{E}_{n}[\{Y_i-m(L_i;\beta_n)\}^2]^{1/2}\\&\quad\times\mathbb{E}_{n}[\{m(L_i;\hat{\beta})-m(L_i;\beta_n)\}^2]^{1/2}
\end{align*}
As before, one can bound $\max_{i\leq n}\{A_i-\pi(L_i;\hat{\gamma})\}^2$, and $\mathbb{E}_{n}[\{m(L_i;\hat{\beta})-m(L_i;\beta_n)\}^2]^{1/2}=o_{P_n}(1)$ by Assumption \ref{p_error}\ref{beta_p_error}. For $\mathbb{E}_{n}[\{Y_i-m(L_i;\beta_n)\}^2]^{1/2}$, note that by Assumption \ref{moment_drt}\ref {m2_drt}, $\mathbb{E}_{P_n}\{|Y_i-m(L_i;\beta_n)|^4\}=O(1)$ and hence $\mathbb{E}_{P_n}[\{Y_i-m(L_i;\beta_n)\}^2]^{1/2}=O(1)$ since the bound on the higher-order moment implies the existence of the lower-order moment. To bound the sample average $\mathbb{E}_{n}[\{Y_i-m(L_i;\beta_n)\}^2]^{1/2}$, by the Von Bahr-Esseen inequality:
\begin{align*}
&\mathbb{P}_{P_n}\bigg(\big|\mathbb{E}_n\{|Y_i-m(L_i;\beta_n)|^4\}-\mathbb{E}_{P_n}\{|Y_i-m(L_i;\beta_n)|^4\}\big|^q>\epsilon\bigg)\\
&\leq \frac{1}{\epsilon^qn^q}\bigg(2-\frac{1}{n}\bigg)\sum^n_{i=1}\mathbb{E}_{P_n}\bigg[\big||Y_i-m(L_i;\beta_n)|^4-\mathbb{E}_{P_n}\{|Y_i-m(L_i;\beta_n)|^4\}\big|^q\bigg]
\end{align*}
for $q\in[1,2]$. Applying Minkowski's inequality and using Assumption \ref{moment_drt}\ref {m2_drt}:
\begin{align*}
&\mathbb{E}_{P_n}\bigg[\big||Y_i-m(L_i;\beta_n)|^4-\mathbb{E}_{P_n}\{|Y_i-m(L_i;\beta_n)|^4\}\big|^q\bigg]\\
&\leq \bigg(\mathbb{E}_{P_n}\left\{|Y_i-m(L_i;\beta_n)|^{4q}\right\}^{1/q}+\mathbb{E}_{P_n}\big[\mathbb{E}_{P_n}\{|Y_i-m(L_i;\beta_n)|^{4}\}^q\big]^{1/q}\bigg)^q\\
&=O(1),
\end{align*}
hence $\mathbb{E}_{n}\{|Y_i-m(L_i;\beta_n)|^4\}=O_{P_n}(1)$ and also 
$\mathbb{E}_{n}[\{Y_i-m(L_i;\beta_n)\}^2]^{1/2}=O_{P_n}(1)$. Therefore $|R_5|=o_{P_n}(1)$.

Similarly, for $R_4$ we have
\begin{align*}
|R_4| \leq &\mathbb{E}_{n}[\{\pi(L_i;\hat{\gamma})-\pi(L_i;\gamma_n)\}^2\{Y_i-m(L_i;\beta_n)\}^2]\\
&+2\max_{i\leq n}|A_i-\pi(L_i;\gamma_n)|\mathbb{E}_{n}\{|Y_i-m(L_i;\beta_n)|^4\}^{1/2}\\&\quad\times\mathbb{E}_{n}[\{\pi(L_i;\hat{\gamma})-\pi(L_i;\gamma_n)\}^2]^{1/2}
\end{align*}
By invoking Assumptions \ref{moment_drt}\ref {m2_drt}, \ref{p_error}\ref{gamma_p_error} and the sparsity condition \ref{S_sum}, one can show that the second term on the right hand side of the inequality is $o_{P_n}(1)$. Regarding the first term, 
\begin{align*}
&\mathbb{E}_{n}[\{\pi(L_i;\hat{\gamma})-\pi(L_i;\gamma_n)\}^2\{Y_i-m(L_i;\beta_n)\}^2]\\
&\leq \max_{i\leq n}|\pi(L_i;\hat{\gamma})-\pi(L_i;\gamma_n)|\mathbb{E}_{n}[\{\pi(L_i;\hat{\gamma})-\pi(L_i;\gamma_n)\}^2]^{1/2}\\&\quad\times\mathbb{E}_{n}\{|Y_i-m(L_i;\beta_n)|^4\}^{1/2}\\
&=o_{P_n}(1)
\end{align*}
using H\"older's inequality, Assumptions \ref{moment_drt}\ref {m2_drt}, \ref{p_error}\ref{gamma_p_error}, the previous result that $\mathbb{E}_{n}\{|Y_i-m(L_i;\beta_n)|^4\}=O_{P_n}(1)$ and the sparsity condition \ref{S_sum}. We have shown (\ref{temp_res}) and result (\ref{Res_3}) follows.

\subsubsection*{Step 4.}\quad

Consider a sequence $P_n \in P$ such that for any $t\in\mathbb{R}$
\[\lim_{n\to \infty}|\mathbb{P}_{P_n}(T_n\leq t)-\Phi(t)|> 0.\]
This directly contradicts the results given above that the test statistic $T_n$ converges to a normal distribution with mean 0 and variance 1 under any subsequence $P_n$ in $P$.
\end{proof}




\subsection{Proof of Theorem \ref{theorem2_drt}}

In the proofs of Theorems \ref{theorem2_drt}  and Theorem \ref{theorem4_drt}, we will rely on some additional assumptions. Specifically, we will make use of the following rates:
\begin{assumption}\label{l_error_w}
(Rates on error of estimated coefficients) Recall that $s^*=s_\gamma \lor s_\beta$; then
\begin{enumerate}[label=(\roman*)]
\item $||\gamma_n-\hat{\gamma}_{BR}||_1=O_{P_n}(s^*\sqrt{\log(p \lor n)/n})$. \label{gamma_l1_error}
\item $||\beta_n-\hat{\beta}_{BR}||_1=O_{P_n}(s^*\sqrt{\log(p \lor n)/n})$. \label{beta_l1_error}
\item $||\gamma_n-\hat{\gamma}_{BR}||_2=O_{P_n}(\sqrt{s^*\log(p \lor n)/n})$. \label{gamma_l2_error}
\item $||\beta_n-\hat{\beta}_{BR}||_2=O_{P_n}(\sqrt{s^*\log(p \lor n)/n})$. \label{beta_l2_error}
\end{enumerate}
\end{assumption}
\begin{assumption}\label{p_error_w}
(Rates on prediction error for weighted estimators) 
\begin{enumerate}[label=(\roman*)]
\item $\mathbb{E}_{n}[\{\pi(L_i;\gamma_n)-\pi(L_i;\hat{\gamma}_{BR})\}^2]=O_{P_n}(s^*\log(p \lor n)/n)$. \label{gamma_br_p_error}
\item $\mathbb{E}_{n}[\{m(L_i;\beta_n)-m(L_i;\hat{\beta}_{BR})\}^2]=O_{P_n}(s^*\log(p \lor n)/n).$ \label{beta_br_p_error}
\end{enumerate}
\end{assumption}
\begin{remark}
\normalfont 
Our proposed nuisance parameter estimators are obtained via (weighted) $\ell_1$ penalized regression. The rates in \ref{l_error_w} and \ref{p_error_w} again follow from the results of \citet{belloni_post-selection_2016} on weighted $\ell_1$-penalized regression (e.g. their Theorem 4); see also \citet{ning_general_2017}. 
\end{remark}

To obtain these rates in Assumptions \ref{l_error_w} and \ref{p_error_w}, we need assumptions on the order of the penalty (which we exploit in the proof); specifically, we need that 
\begin{align}
\lambda_{\gamma}&=O\bigg(\sqrt{\frac{\log (p\lor n)}{n}}\bigg)\label{penalty_assumption_gamma}\\ \lambda_{\beta}&=O\bigg(\sqrt{\frac{\log (p\lor n)}{n}}\bigg)\label{penalty_assumption_beta}. 
\end {align}
These are standard assumption on the order of the penalty level in the literature, when working either under the intersection submodel \citep{farrell_robust_2015,belloni_post-selection_2016}, or the union model \citep{avagyan_honest_2017}. 

\begin{proof}
Repeating the previous decomposition of \\$\sqrt{n}\mathbb{E}_n\{U_i(\hat{\eta}_{BR})\}-\sqrt{n}\mathbb{E}_n\{U_i(\eta_n)\}$, for $R_1$, we now have 
\begin{align}
&\frac{1}{\sqrt{n}}\sum^n_{i=1}\{A_i-\pi(L_i;\gamma_n)\}\{m(L_i;\beta_n)-m(L_i;\hat{\beta}_{BR})\}\nonumber\\
&=\frac{1}{\sqrt{n}}\sum^n_{i=1}\{A_i-\pi(L_i;\hat{\gamma}_{BR})\}\{m(L_i;\beta_n)-m(L_i;\hat{\beta}_{BR})\}\label{R11}\\
&\quad+\frac{1}{\sqrt{n}}\sum^n_{i=1}\{\pi(L_i;\hat{\gamma}_{BR})-\pi(L_i;\gamma_n)\}\{m(L_i;\beta_n)-m(L_i;\hat{\beta}_{BR})\}\label{R12}
\end{align}
Then for (\ref{R11}), note that following a Taylor expansion, 
\begin{align*}
&\frac{1}{\sqrt{n}}\sum^n_{i=1}\{A_i-\pi(L_i;\hat{\gamma}_{BR})\}\{m(L_i;\beta_n)-m(L_i;\hat{\beta}_{BR})\}
\\&=-\sqrt{n}\mathbb{E}_n\bigg\{\frac{\partial U_i(\hat{\eta}_{BR})}{\partial \beta}\bigg\} (\beta_n-\hat{\beta}_{BR})+O_{P_n}(\sqrt{n}||\beta_n-\hat{\beta}_{BR}||^2_2)
\end{align*}
and by H\"older's inequality,
\begin{align*}
&\bigg|\sqrt{n}\mathbb{E}_n\bigg\{\frac{\partial U_i(\hat{\eta}_{BR})}{\partial \beta}\bigg\} (\beta_n-\hat{\beta}_{BR})\bigg|\\
&\leq\sqrt{n}\lambda_\gamma||\beta_n-\hat{\beta}_{BR}||_1
\end{align*}
using the stationarity conditions for the $\ell_1$-penalised estimator of $\gamma_0$.
Therefore, given (\ref{penalty_assumption_gamma}), Assumptions \ref{l_error_w}\ref{beta_l1_error}, \ref{l_error_w}\ref{beta_l2_error} and sparsity condition \ref{US_sum},
\[\bigg|\frac{1}{\sqrt{n}}\sum^n_{i=1}\{A_i-\pi(L_i;\hat{\gamma}_{BR})\}\{m(L_i;\beta_n)-m(L_i;\hat{\beta}_{BR})\}\bigg|=o_{P_n}(1)\]
Considering the other term (\ref{R12}), along the same lines as in the proof of Theorem \ref{theorem1_drt},
one can show that 
\begin{align*}
&\bigg|\frac{1}{\sqrt{n}}\sum^n_{i=1}\{\pi(L_i;\hat{\gamma}_{BR})-\pi(L_i;\gamma_n)\}\{m(L_i;\beta_n)-m(L_i;\hat{\beta}_{BR})\}\bigg|=o_{P_n}(1)
\end{align*}
using H\"older's inequality, sparsity condition \ref{US_sum} and Assumptions \ref{p_error_w}\ref{gamma_br_p_error} and \ref{p_error_w}\ref{beta_br_p_error}. Therefore $|R_1|=o_{P_n}(1)$. One can re-apply the argument given immediately above to show that $|R_3|=o_{P_n}(1)$.

By noting that
\begin{align*}
&\frac{1}{\sqrt{n}}\sum^n_{i=1}\{Y_i-m(L_i;\hat{\beta}_{BR})\}\{\pi(L_i;\gamma_n)-\pi(L_i;\hat{\gamma}_{BR})\}
\\&=-\sqrt{n}\mathbb{E}_n\bigg\{\frac{\partial U_i(\hat{\eta}_{BR})}{\partial \gamma}\bigg\} (\gamma_n-\hat{\gamma}_{BR}) +O_{P_n}(\sqrt{n}||\gamma_n-\hat{\gamma}_{BR}||^2_2)
\end{align*}
and
\begin{align*}
&\bigg|\sqrt{n}\mathbb{E}_n\bigg\{\frac{\partial U_i(\hat{\eta}_{BR})}{\partial \gamma}\bigg\}  (\gamma_n-\hat{\gamma}_{BR})\bigg|\\
&\leq\sqrt{n}\lambda_\beta||\gamma_n-\hat{\gamma}_{BR}||_1
\end{align*}
one can also repeat the above arguments to show that $|R_2|=o_{P_n}(1)$, given (\ref{penalty_assumption_beta}), Assumptions  \ref{l_error_w}\ref{beta_l1_error}, \ref{l_error_w}\ref{beta_l2_error}, \ref{p_error_w}\ref{gamma_br_p_error}, \ref{p_error_w}\ref{beta_br_p_error} and sparsity condition \ref{US_sum}. Result (\ref{Res_1}) follows immediately and the main result follows by essentially repeating Steps 2-4 from the proof of Theorem \ref{theorem1_drt}.
\end{proof}

\subsection{Proof of Corollary \ref{theorem3_drt}}

\begin{proof}
As discussed in the main paper, when $E(A|L)=\pi(L;\gamma^*)$ and $\gamma^*$ is known, the proposal in Section \ref{nuis_mis} for estimating $\beta$ reduces to standard (unweighted) PMLE.  Then repeating Step 1 of the previous proof,
\begin{align*}
&\frac{1}{\sqrt{n}}\sum^n_{i=1}\{A_i-\pi(L_i;\gamma^*)\}\{Y_i-m(L_i;\hat{\beta})\}\\&=\frac{1}{\sqrt{n}}\sum^n_{i=1}\{A_i-\pi(L_i;\gamma^*)\}\{Y_i-m(L_i;\beta_n)\}\\
&\quad+\frac{1}{\sqrt{n}}\sum^n_{i=1}\{A_i-\pi(L_i;\gamma^*)\}\{m(L_i;\beta_n)-m(L_i;\hat{\beta})\}
\end{align*}
Let us define $R_1^*=\sqrt{n}\mathbb{E}_n[\{A_i-\pi(L_i;\gamma^*)\}\{m(L_i;\beta_n)-m(L_i;\hat{\beta})\}]$.

Then,
\begin{align*}
&\mathbb{E}_{P_n}\{R^{*2}_{1}|(Y_i,L_i)^n_{i=1}\}\\&=\mathbb{E}_{n}\bigg(\mathbb{E}_{P_n}[\{A_i-\pi(L_i;\gamma^*)\}^2|(L_i)^n_{i=1}]\{m(L_i;\beta_n)-m(L_i;\hat{\beta})\}^2\bigg)
\\&\leq C\mathbb{E}_{n}[\{m(L_i;\beta_n)-m(L_i;\hat{\beta})\}^2]
\end{align*}
where $C$ is a constant. Invoking \ref{p_error_w}\ref{beta_p_error} and sparsity condition \ref{S_sum}, we have 
\[C\mathbb{E}_{n}[\{m(L_i;\beta_n)-m(L_i;\hat{\beta})\}^2]=o_{P_n}(1),\]
hence $\mathbb{E}_{P_n}[R^{*2}_1]=o(1)$ and $|R^*_1|=o_{P_n}(1)$ using Chebyshev's Inequality. Note that sparsity condition \ref{US_sum} has not been invoked.  
\end{proof}

\subsection{Auxiliary results on weighted estimators}\label{lemma_app}

Here, we restrict to settings where $m(L;\beta_n)=\beta^T_nL$ and $\pi(L;\gamma_n)=\expit(\gamma^T_nL)$. In this case, note that $\hat{\gamma}_{BR}=\hat{\gamma}$, since no weights are used in estimating this parameter. 
In order to make transparent the dependence of the estimator $\hat{\beta}_{BR}$ on the weights, we introduce the notation  $\hat{\beta}(\hat{\gamma})$ for when $\gamma$ (required for the weights) is estimated from the data and $\hat{\beta}(\gamma_n)$ otherwise. 

In what follows, we will give a lemma  (key to proving Theorem \ref{theorem4_drt}) regarding the quantity $\hat{\beta}(\gamma_n)-\hat{\beta}(\hat{\gamma})$. This will be helpful for understanding the impact of using estimated weights on the distribution of the test statistic. Several additional assumptions are required:
\begin{assumption}\label{strong_cons}
(Fast convergence of estimating equations)
\begin{align}
&\norm{\mathbb{E}_n[w(L_i;\gamma_n)[Y_i-m\{L_i;\hat{\beta}(\gamma_n)\}]L_i-w(L_i;\hat{\gamma})[Y_i-m\{L_i;\hat{\beta}(\hat{\gamma})\}]L_i]}_\infty\nonumber\\
&=O_{P_n}\bigg(\frac{\sqrt{\log p}}{n}\bigg)\label{strong_cons_eq}
\end{align}
\end{assumption}
\begin{remark}
\normalfont This assumption requires the difference between the estimating equations for $\hat{\beta}(\hat{\gamma})$ and $\hat{\beta}(\gamma_n)$ to shrink very quickly. In low-dimensional settings, this difference is exactly zero by virtue of the estimation procedure. In the high-dimensional setting, even stronger results than (\ref{strong_cons_eq}) are available for our proposed estimators, if we represent our estimator of $\beta_0$ as the solution to estimating equations with a bridge penalty:
\begin{align*}
&0=\frac{1}{n}\sum^n_{i=1}\frac{\partial}{\partial \gamma}U_i(\hat{\eta}_{BR})+\lambda_\beta\delta|\hat{\beta}_{BR}|^{\delta-1}\circ\sign(\hat{\beta}_{BR}).
\end{align*}
Here, $\circ$ is the Hadamard product operator. As $\delta\to 1+$, the penalty term corresponds to the subgradient of the $\ell_1$ or Lasso norm penalty $||\beta||_1$ with respect to $\beta$. Using this representation of our estimator, we can see that
\begin{align*}
&\mathbb{E}_n[w(L_i;\gamma_n)[Y_i-m\{L_i;\hat{\beta}(\gamma_n)\}]L_i-w(L_i;\hat{\gamma})[Y_i-m\{L_i;\hat{\beta}(\hat{\gamma})\}]L_i]\\&
=\sqrt{n}\lambda_\beta \delta[|\hat{\beta}(\gamma_n)|^{\delta-1}\circ\sign(\hat{\beta}(\gamma_n))-|\hat{\beta}(\hat{\gamma})|^{\delta-1}\circ\sign(\hat{\beta}(\hat{\gamma}))]\\
&=\sqrt{n}\lambda_\beta \delta(\delta-1)|\hat{\beta}(\hat{\gamma})|^{\delta-2}\circ\sign(\hat{\beta}(\hat{\gamma}))\{\hat{\beta}(\gamma_n)-\hat{\beta}(\hat{\gamma})\}\\&\quad+O_{P_n}(\sqrt{n}||\hat{\beta}(\gamma_n)-\hat{\beta}(\hat{\gamma})||^2_2) 
\end{align*}
where the final equality follows from a Taylor expansion around $\hat{\beta}(\hat{\gamma})$. For any finite $n$, we can choose $\delta$ to be close enough to 1 such that 
\[\sqrt{n}\lambda_\beta \delta(\delta-1)|\hat{\beta}(\hat{\gamma})|^{\delta-2}\circ\sign(\hat{\beta}(\hat{\gamma}))\{\hat{\beta}(\gamma_n)-\hat{\beta}(\hat{\gamma})\}\]
is negligible, since $\hat{\beta}(\gamma_n)-\hat{\beta}(\hat{\gamma})$ is assumed not to diverge as $\delta\to 1+$. 
\end{remark}
\begin{assumption}\label{hlms}
(High-dimensional model selection)\\ Let $\hat{s}_{\beta}$ denote the cardinality $|\support\{\hat{\beta}(\hat{\gamma})\}\cup\support\{\hat{\beta}(\gamma_n)\}|$; then $\hat{s}_{\beta}=O_{P_n}(s^*)$.
\end{assumption}
\begin{remark}
\normalfont This assumption states that the number of non-zero entries common to both $\hat{\beta}(\hat{\gamma})$ and $\hat{\beta}(\gamma_n)$ is of similar order to $s^*$, and is satisfied when $|\support\{\hat{\beta}(\hat{\gamma})\}|=O_{P_n}(s^*)$; note that it does not require perfect model selection \citep{belloni_inference_2014}.
\end{remark}

Before giving Lemma \ref{lemma1}, we will also review several regularity conditions necessary for the consistency of $\ell_1$-penalized estimators. This is done for clarity, since they will be invoked in the following proofs. For example, we will use the concentration bound 
\begin{align}\label{concentration_drt}
\norm{\mathbb{E}_n[\{A_i-\pi(L_i;\gamma_n)\}L_i]}_\infty=O_{P_n}\left(\sqrt{\frac{\log (p\lor n)}{n}}\right);
\end{align}
this can be shown to hold either using the theory of moderate deviations for self-normalised sums \citep{de_la_pena_self-normalized_2009,belloni_sparse_2012}, or using sub-Gaussian primitive conditions. Furthermore, restricted eigenvalues conditions are also required. For the estimator $\hat{\beta}(\hat{\gamma})$ we need that for any constant $\zeta\geq1$, there exists a finite constant $\kappa>0$ such that 
\begin{align}
\min \left[\frac{\hat{\delta}^TM\hat{\delta}}{\norm{\hat{\delta}}^2_2}: \hat{\delta}\in\mathbb{R}^p, \norm{\hat{\delta}_{T^C}}_1\leq \zeta\norm{\hat{\delta}_{T}}_1,\hat{\delta}\neq0  \right] &\geq \kappa > 0 \label{RE_hat} 
\end{align}
where $\hat{\delta}=\hat{\beta}(\hat{\gamma})-\beta_n$ and $M=\mathbb{E}_n\{L_iw(L_i;\gamma_n)L_i^T\}$.
\begin{remark}
\normalfont 
If the cone constraint on $\hat{\delta}$ in (\ref{RE_hat}) is satisfied (and the equivalent condition is satisfied for $\tilde{\delta}=\hat{\beta}(\gamma_n)-\beta_n$), then by the triangle inequality,
\begin{align*}
\norm{\hat{\beta}(\hat{\gamma})_{T^C}-\hat{\beta}(\gamma_n)_{T^C}}_1&= \norm{\hat{\beta}(\hat{\gamma})_{T^C}  -\beta_{nT^C}+\beta_{nT^C}-    \hat{\beta}(\gamma_n)_{T^C}}_1\\
&\leq \norm{\hat{\delta}_{T^C}}_1+\norm{\tilde{\delta}_{T^C}}_1\\
&\leq \zeta \left(\norm{\hat{\delta}_{T}}_1+\norm{\tilde{\delta}_{T}}_1\right).
\end{align*}
Defining the set 
\[\Delta_{\zeta,T}=\left\{\hat{\beta}(\hat{\gamma})-\hat{\beta}(\gamma_n)\in\mathbb{R}^p \backslash \{0\}: \norm{\hat{\beta}(\hat{\gamma})_{T^C}-\hat{\beta}(\gamma_n)_{T^C}}_1\leq \zeta \left(\norm{\hat{\delta}_{T}}_1+\norm{\tilde{\delta}_{T}}_1\right)\right\}\]
and the minimum restricted eigenvalue of $M$ as
\begin{align*}
\varphi^2_\zeta(M)=\min_{\hat{\beta}(\hat{\gamma})-\hat{\beta}(\gamma_n) \in\Delta_{\zeta,T},|T|\leq s^*} \left[ \frac{\{\hat{\beta}(\hat{\gamma})-\hat{\beta}(\gamma_n)\}^TM\{\hat{\beta}(\hat{\gamma})-\hat{\beta}(\gamma_n)\}}{\norm{\hat{\beta}(\hat{\gamma})-\hat{\beta}(\gamma_n)}^2_2}\right],
\end{align*}
condition (\ref{RE_hat}) thus implies that 
\begin{align}\label{RE_lb}
\varphi^2_\zeta(M)\geq \kappa > 0.
\end{align}
\end{remark}
For a comprehensive discussion of suitable regularity conditions for proving the consistency of $\ell_1$-penalized estimators, we refer the interested reader to \citet{buhlmann_statistics_2011}.

\begin{lemma}\label{lemma1}
In addition to the sparsity conditions \ref{S_sum} and \ref{S_prod_br}, suppose that the Assumptions \ref{moment_drt}\ref{m3_drt}, \ref{p_error}\ref{gamma_p_error}, \ref{l_error_w}\ref{gamma_l1_error}, \ref{p_error_w}\ref{beta_br_p_error}, \ref{strong_cons} and \ref{hlms} hold. Then it follows that for a given sequence $P_n$ we have that
\begin{align}\label{lemma_res}
&\norm{\hat{\beta}(\hat{\gamma})-\hat{\beta}(\gamma_n)}_1\nonumber \\&=O_{P_n}\left(\max_{i\leq n}|\epsilon_i|\sqrt{\frac{s_\gamma s^* \log (p\lor n)}{n}}+\frac{s^*\sqrt{ \log (p\lor n)}( s_\gamma \sqrt{ \log (p\lor n)}+1)}{n}\right)
\end{align}
where $\epsilon=Y-m(L;\beta_n)$.
\end{lemma}

\begin{proof}
Many of the arguments are similar to those in Appendix E of \citet{ning_general_2017}. Let $\bar{\delta}=\hat{\beta}(\hat{\gamma})-\hat{\beta}(\gamma_n)$; then 
\begin{align*}
\bar{\delta}^TM\bar{\delta}
&=\mathbb{E}_n\left(w(L_i;\gamma_n)(\bar{\delta}^TL_i)[Y_i-m\{L_i;\hat{\beta}(\gamma_n)\}]\right)\\&\quad
-\mathbb{E}_n\left(w(L_i;\gamma_n)(\bar{\delta}^TL_i)[Y_i-m\{L_i;\hat{\beta}(\hat{\gamma})\}]\right)\\
&=\mathbb{E}_n\left(w(L_i;\gamma_n)(\bar{\delta}^TL_i)[Y_i-m\{L_i;\hat{\beta}(\gamma_n)\}]\right)\\
&\quad-\mathbb{E}_n\left(w(L_i;\hat{\gamma})(\bar{\delta}^TL_i)[Y_i-m\{L_i;\hat{\beta}(\hat{\gamma})\}]\right)\\
&\quad +\mathbb{E}_n\left[\{w(L_i;\hat{\gamma})-w(L_i;\gamma_n)\}(\bar{\delta}^TL_i)\{Y_i-m(L_i;\beta_n)\}\right]\\ 
&\quad +\mathbb{E}_n\left(\{w(L_i;\hat{\gamma})-w(L_i;\gamma_n)\}(\bar{\delta}^TL_i)[m(L_i;\beta_n)-m\{L_i;\hat{\beta}(\hat{\gamma})\}]\right)\\
&=R_6+R_7+R_8+R_9.
\end{align*}

Firstly, 
\begin{align*}
&|R_6+R_7|\\
&\leq \norm{\mathbb{E}_n[w(L_i;\gamma_n)\{Y_i-\tilde{m}(L_i)\}L_i-w(L_i;\hat{\gamma})\{Y_i-\hat{m}(L_i)\}L_i]}_\infty\norm{\bar{\delta}}_1\\
&\leq C'\frac{\sqrt{\log (p\lor n)}}{n}\norm{\bar{\delta}}_1
\end{align*}
where $C'$ is a constant, due to Assumption \ref{strong_cons}. For $R_8$, 
\begin{align*}
|R_8|&=\bigg|\mathbb{E}_n\left[\left\{\frac{w(L_i;\hat{\gamma})-w(L_i;\gamma_n)}{w(L_i;\gamma_n)(\hat{\gamma}^TL_i-\gamma^T_nL_i)}\right\}\sqrt{w(L_i;\gamma_n)}(\bar{\delta}^TL_i)\right.\\&\left.\quad \times  \sqrt{w(L_i;\gamma_n)}(\hat{\gamma}^TL_i-\gamma^T_nL_i)\{Y_i-m(L_i;\beta_n)\}\right]\bigg|\\
&\leq \max_{i\leq n}|\epsilon_i|(\bar{\delta}^TM\bar{\delta})^{1/2}\mathbb{E}_n\left[(\hat{\gamma}^TL_i-\gamma^T_nL_i)^2\right]^{1/2}\\
&\lesssim \max_{i\leq n}|\epsilon_i|(\bar{\delta}^TM\bar{\delta})^{1/2}\sqrt{\frac{s_\gamma \log (p\lor n)}{n}}
\end{align*}
under Assumption \ref{p_error} and using the properties of the logistic link function. 
Next, 
\begin{align*}
|R_9|
&\leq \mathbb{E}_n[w(L_i;\gamma_n)|\hat{\gamma}^TL_i-\gamma^T_nL_i||\bar{\delta}^TL_i||m(L_i;\beta_n)-m\{L_i;\hat{\beta}(\hat{\gamma})\}|]\\
&\leq \max_{i\leq n}\norm{L_i}_\infty\norm{\hat{\gamma}-\gamma_n}_1 \mathbb{E}_n[w(L_i;\gamma_n)|\bar{\delta}^TL_i||m(L_i;\beta_n)-m\{L_i;\hat{\beta}(\hat{\gamma})\}|]\\
&\leq \max_{i\leq n}\norm{L_i}_\infty \norm{\hat{\gamma}-\gamma_n}_1 \\&\quad \times (\bar{\delta}^TM\bar{\delta})^{1/2}\mathbb{E}_n\left([m(L_i;\beta_n)-m\{L_i;\hat{\beta}(\hat{\gamma})\}]^2\right)^{1/2}\\
&\lesssim (\bar{\delta}^TM\bar{\delta})^{1/2} \left(\frac{s_\gamma\sqrt{s^*}\log (p\lor n)}{n}\right)
\end{align*}
invoking Assumptions \ref{moment_drt}\ref{m3_drt}, \ref{l_error_w}\ref{gamma_l1_error}, \ref{p_error_w} and \ref{beta_br_p_error}.

Putting this together, 
\begin{align}\label{comb_res}
\bar{\delta}^TM\bar{\delta}\leq&  C'\frac{\sqrt{\log p}}{n}\norm{\bar{\delta}}_1\nonumber\\&+(\bar{\delta}^TM\bar{\delta})^{1/2}C^*\left(\max_{i\leq n}|\epsilon_i|\sqrt{\frac{s_\gamma \log (p\lor n)}{n}}+\frac{s_\gamma\sqrt{s^*}\log (p\lor n)}{n}\right)
\end{align}
where $C^*$ is a constant. Let us consider two cases; firstly, assume
\[(\bar{\delta}^TM\bar{\delta})^{1/2}\leq C^*\left(\max_{i\leq n}|\epsilon_i|\sqrt{\frac{s_\gamma \log (p\lor n)}{n}}+\frac{s_\gamma\sqrt{s^*}\log (p\lor n)}{n}\right)\]
holds. Note (\ref{RE_lb}) implies that 
\begin{align*}
(\bar{\delta}^TM\bar{\delta})^{1/2}&\geq \varphi_\zeta(M)\norm{\bar{\delta}}_2\gtrsim 
\frac{1}{\sqrt{\hat{s}_{\beta}}}\norm{\bar{\delta}}_1.
\end{align*}
Invoking Assumption \ref{hlms} and combining the lower and upper bound,
\begin{align}\label{1st_case_res}
\norm{\bar{\delta}}_1 \lesssim \max_{i\leq n}|\epsilon_i|\sqrt{\frac{s_\gamma s^*\log (p\lor n)}{n}}+\frac{s_\gamma s^*\log (p\lor n)}{n}
\end{align}
On the other hand, if 
\[(\bar{\delta}^TM\bar{\delta})^{1/2}\geq C^*\left(\max_{i\leq n}|\epsilon_i|\sqrt{\frac{s_\gamma \log (p\lor n)}{n}}+\frac{s_\gamma\sqrt{s^*}\log (p\lor n)}{n}\right)\]
then rearranging (\ref{comb_res}), it follows that 
\begin{align*}
&(\bar{\delta}^TM\bar{\delta})^{1/2}\left\{(\bar{\delta}^TM\bar{\delta})^{1/2}-C^*\left(\max_{i\leq n}|\epsilon_i|\sqrt{\frac{s_\gamma \log (p\lor n)}{n}}+\frac{s_\gamma\sqrt{s^*}\log (p\lor n)}{n}\right)\right\}\\
&\leq C'\frac{\sqrt{\log (p\lor n)}}{n}\norm{\bar{\delta}}_1
\end{align*}
Using (\ref{RE_lb}), we have 
\begin{align*}
&(\bar{\delta}^TM\bar{\delta})^{1/2}-C^*\left(\max_{i\leq n}|\epsilon_i|\sqrt{\frac{s_\gamma \log (p\lor n)}{n}}+\frac{s_\gamma\sqrt{s^*}\log (p\lor n)}{n}\right)\\&\lesssim \frac{\sqrt{s^* \log (p\lor n)}}{n}
\end{align*}
so 
\begin{align*}
&(\bar{\delta}^TM\bar{\delta})^{1/2}\lesssim \max_{i\leq n}|\epsilon_i|\sqrt{\frac{s_\gamma \log (p\lor n)}{n}}+\frac{s_\gamma\sqrt{s^*}\log (p\lor n)+\sqrt{s^* \log (p\lor n)}}{n}
\end{align*}
and again by (\ref{RE_lb}),
\begin{align}\label{2nd_case_res}
&\norm{\bar{\delta}}_1\lesssim \max_{i\leq n}|\epsilon_i|\sqrt{\frac{s_\gamma s^* \log (p\lor n)}{n}}+\frac{s^*\sqrt{ \log (p\lor n)}( s_\gamma \sqrt{ \log (p\lor n)}+1)}{n}.
\end{align}
Taking the union of the bounds (\ref{1st_case_res}) and (\ref{2nd_case_res}) completes the proof.
\end{proof}

\subsection{Proof of Theorem \ref{theorem4_drt}}

\begin{assumption}\label{bounded_res}
(Regularity conditions on the errors)\\ $\max_{i\leq n}|\epsilon_i|\sqrt{s_\gamma s^*}\log(p\lor n)=o(\sqrt{n})$ w.p. 1.
\end{assumption}
\begin{remark}
\normalfont This can be most simply shown to hold if $\max_{i\leq n}|\epsilon_i|=O_{P_n}(1)$. \citet{belloni_sparse_2012} and \citet{farrell_robust_2015} suppose that $\max_{i\leq n}|\epsilon_i|=O_{P_n}(n^{1/r})$ for some $r>2$, such that larger values of $r$ allow one to relax assumptions on sparsity in exchange for stronger conditions on the distributions of the errors. If the $\epsilon_i$ are normal, then $r$ can be arbitrarily large. Alternatively, one can place the stronger sub-Gaussian conditions on $\epsilon_i$, whereby $\max_{i\leq n}|\epsilon_i|=O_{P_n}(\sqrt{\log n})$. 
\end{remark}

We now give the proof.

\begin{proof}

Decomposing $\sqrt{n}\mathbb{E}_n\{U_i(\hat{\eta}_{BR})\}-\sqrt{n}\mathbb{E}_n\{U_i(\eta_n)\}$
as in the proof of Theorem \ref{theorem1_drt}, one can show $|R_2|=o_{P_n}(1)$ along the lines of the proof of Theorem \ref{theorem1_drt}, appealing to Assumptions \ref{moment_drt}\ref{m1_drt}, \ref{p_error}\ref{gamma_p_error} and sparsity condition \ref{S_sum}. Similarly, one can show that $R_3$ is $o_{P_n}(1)$ using the joint sparsity condition \ref{S_prod} and Assumptions \ref{p_error}\ref{gamma_p_error} and \ref{p_error_w}\ref{beta_br_p_error}.

Then for $R_1$,
\begin{align*}
&\frac{1}{\sqrt{n}}\sum^n_{i=1}\{A_i-\pi(L_i;\gamma_n)\}[m(L_i;\beta_n)-m\{L_i;\hat{\beta}(\hat{\gamma})\}]\\
&=\frac{1}{\sqrt{n}}\sum^n_{i=1}\{A_i-\pi(L_i;\gamma_n)\}[m(L_i;\beta_n)-m\{L_i;\hat{\beta}(\gamma_n)\}]\\
&\quad+\frac{1}{\sqrt{n}}\sum^n_{i=1}\{A_i-\pi(L_i;\gamma_n)\}[m\{L_i;\hat{\beta}(\gamma_n)\}-m\{L_i;\hat{\beta}(\hat{\gamma})\}]\\
&=R_{1a}+R_{1b}
\end{align*}
One can show $|R_{1a}|=o_{P_n}(1)$ using Assumption \ref{p_error_w}\ref{beta_br_p_error} and sparsity condition \ref{S_sum}. For $R_{1b}$, 
\begin{align*}
&\bigg|\frac{1}{\sqrt{n}}\sum^n_{i=1}\{A_i-\pi(L_i;\gamma_n)\}[m\{L_i;\hat{\beta}(\gamma_n)\}-m\{L_i;\hat{\beta}(\hat{\gamma})\}]\bigg|
\\
&\leq\sqrt{n}\norm{\mathbb{E}_n[\{A_i-\pi(L_i;\gamma_n)\}L_i]}_{\infty} \norm{\hat{\beta}(\gamma_n)-\hat{\beta}(\hat{\gamma})}_1\\
&=O_{P_n}(\sqrt{\log (p\lor n)})\\&\quad \times O_{P_n}\left( \max_{i\leq n}|\epsilon_i|\sqrt{\frac{s_\gamma s^* \log (p\lor n)}{n}}+\frac{s^*\sqrt{ \log (p\lor n)}( s_\gamma \sqrt{ \log (p\lor n)}+1)}{n}\right)
\end{align*}
by (\ref{concentration_drt}) and Lemma 1. Hence under sparsity assumptions \ref{S_sum}, \ref{S_prod_br} and Assumption \ref{bounded_res}, $|R_{1b}|=o_{P_n}(1)$ and thus $R_1=o_{P_n}(1)$. By then repeating Steps 2-4 from the proof of Theorem \ref{theorem1_drt}, the main result follows.
\end{proof}

\begin{remark}
\normalfont We note that it follows from the above proof that sharper results are also available in linear models under misspecification than are given in Theorem \ref{theorem2_drt}. Namely, when either model $\mathcal{A}$ or the linear model for $\mathcal{B}$ is misspecified, ultra-sparsity is only required in the correct model. For example, if $E(A|L)=\pi(L;\gamma)$, then we require $s^2_\gamma=o(n)$ but only $s_\beta=o(n)$ (ignoring log factors).
\end{remark}

\section{Appendix B}\label{appB_drt}

In the figure on the following page, we describe an iterative method for $\gamma$ and $\beta$ (based on the reasoning in Section \ref{nuis_mis}), when both are parameters indexing logistic models. In practice, one can take the penalty terms obtained via cross validation during the first iteration of the algorithm ($j=1$) and use the same terms in subsequent iterations.

\begin{algorithm}
\caption{An algorithm for estimating $\eta$ when $Y$ is binary}\label{alg:euclid}
\begin{enumerate}
\item Estimate $\gamma$ and $\beta$ as $\hat{\gamma}^{(0)}$ and $\hat{\beta}^{(0)}$ using (unweighted) $\ell_1$-penalized logistic regression. Let $\check{\gamma}^{(0)}$ and $\check{\beta}^{(0)}$ denote the refitted estimates.
\item Calculate the weights $w(L_i;\hat{\gamma}^{(0)})=\expit(\hat{\gamma}^{(0)'}L_i)\{1-\expit(\hat{\gamma}^{(0)'}L_i)\}$, $w(L_i;\hat{\beta}^{(0)})=\expit(\hat{\beta}^{(0)'}L_i)\{1-\expit(\hat{\beta}^{(0)'}L_i)\}$, $w(L_i;\check{\gamma}^{(0)})$ and $w(L_i;\check{\beta}^{(0)})$. Calculate the objective function
\begin{align*}
\check{\nu}^{(0)}=&\frac{1}{n}\sum^{n}_{i=1}\log\{1+\exp(\check{\gamma}^{(0)'}L_i)\}-A_i(\check{\gamma}^{(0)'}L_i)+\log\{1+\exp(\check{\beta}^{(0)'}L_i)\}\\&\quad-Y_i(\check{\beta}^{(0)'}L_i)
\end{align*}
\item Set $j=0$ and carry out the following recursive algorithm:
\begin{enumerate}
\item Set $j=j+1$.
 \item Using the initial estimates, re-estimate $\gamma$ and $\beta$ as the solutions $\hat{\gamma}^{(j)}$ and $\hat{\beta}^{(j)}$ to
\begin{align*}
0&=\sum^n_{i=1}w(L_i;\hat{\beta}^{(j-1)'})\{A_i-\expit(\gamma^TL_i)\}L_i+\lambda_\gamma\delta|\gamma|^{\delta-1}\circ\sign(\gamma)\\
0&=\sum^n_{i=1}w(L_i;\hat{\gamma}^{(j-1)'})\{Y_i-\expit(\beta^TL_i)\}L_i+\lambda_\beta\delta|\beta|^{\delta-1}\circ\sign(\beta)
\end{align*} 
Similarly, using $w(L_i;\check{\gamma}^{(j-1)})$ and $w(L_i;\check{\beta}^{(j-1)})$, obtain the refitted $\check{\gamma}^{(j)}$ and $\check{\beta}^{(j)}$.
\item Re-evaluate the objective function as:
\begin{align*}
\check{\nu}^{(j)}=&\frac{1}{n}\sum^{n}_{i=1}\bigg[\log\{1+\exp(\check{\gamma}^{(j)'}L_i)\}-A_i(\check{\gamma}^{(j)'}L_i)\bigg]w(L_i;\check{\beta}^{(j-1)'})\\
&\quad\quad+\bigg[\log\{1+\exp(\check{\beta}^{(j)'}L_i)\}-Y_i(\check{\beta}^{(j)'}L_i)\bigg]w(L_i;\check{\gamma}^{(j-1)'})
\end{align*}
\item If $|\check{\nu}^{(j)}-\check{\nu}^{(j-1)}|<0.0001$, stop the algorithm, and set $\check{\gamma}_{BR}=\check{\gamma}^{(j)}$ and $\check{\beta}_{BR}=\check{\beta}^{(j)}$.
\end{enumerate}
\end{enumerate}
\end{algorithm}

\section{Appendix C}\label{appC_drt}

Here we include some additional simulation results. Compared with the setting considered in Section \ref{simulations} of the main paper, we allowed for a more dense model for the exposure $A$. Specifically $g=\left(40\frac{\log(20)}{n^{1/2}}, 40 \frac{\log(19)}{n^{1/2}}, ...,40 \frac{\log(2)}{n^{1/2}},0_{20},0_{21}, 2 \frac{\log(2)}{n^{1/2}},...,2\frac{\log(60)}{n^{1/2}},0_{81},...,0_{p}  \right)$ now. Data generation process under the alternative outcome model (to evaluate the impact of fitting a misspecified linear model) was the same as in the main paper. 

We further relaxed the assumption in Section \ref{simulations} of homoscedastic errors. Specifically, the outcome $Y_i$ was normally distributed conditional on $L_i$ with mean $1 + \beta_iL_i$ and standard deviation $\sqrt{\left( \dfrac{(1 + \beta_iL_i)^2}{E_n((1 + \beta_iL_i)^2)} \right)}$. In settings where we fitted a misspecified outcome model $Y_i$ was normally distributed with mean $1+\beta^T\left(|L_{.,[1:3]}|; L_{.,[4:p]}\right)$ and standard deviation  \\$\sqrt{\left( \dfrac{\{1+\beta^T\left(|L_{.,[1:3]}|; L_{.,[4:p]}\right)\}^2}{E_n[\{1+\beta^T\left(|L_{.,[1:3]}|; L_{.,[4:p]}\right)\}^2]} \right)}$.

\begin{table}[h]
	\centering
	\caption{ Type I errors based on 1,000 replications  in settings with a denser propensity score: $\Sigma=\textnormal{I}_{p\times p}$.}
	\label{Table1app}
	\begin{tabular}{l  r r r r r r r}
		\\
		\textit{Correct models} &&\\
		
		\hline
		& $n=200$   &   &  $n=500$  & & $n=200$ & & $n=200$\\  
		Methods    &  $p=200$ &    &  $p=500$ & & $p=100$ & &  $p=250$ \\ \hline
		Standard na\"ive (forced) & 0.520   & &  0.813  & &  0.289   && 0.602   \\  
		Standard na\"ive (not forced) & 0.272  & &  0.522  & &  0.171   && 0.336   \\  
		PDS (pre-specified) & 0.526  & & 0.779   & & 0.523    &&  0.517  \\  
		PO (pre-specified) & 0.517  & &  0.758  & &  0.521   &&  0.502  \\  
		PDS (CV) & 0.074  & & 0.070   & &  0.066   &&  0.079  \\  
		PMLE-DR & 0.061  & & 0.053   & &  0.067  &&  0.077  \\

		BR-DR  & 0.054  & & 0.061   & &  0.064   && 0.054  \\  \hline\\

		\textit{Incorrect outcome model} &&\\
		\hline
		& $n=200$  & & $n=500$  && $n=200$ && $n=200$ \\ 
		Methods  &  $p=200$ &  &  $p=500$  && $p=100$ && $p=250$\\ \hline
		Standard na\"ive (forced) & 0.363  & &  0.576  & & 0.205    && 0.456   \\  
		Standard na\"ive (not forced) & 0.169  & &  0.304   & &  0.114   &&  0.196  \\  
		PDS (pre-specified) & 0.344   & &  0.611  & &  0.342   &&  0.310  \\  
		PO (pre-specified)  & 0.338   & &  0.596  & &  0.325   &&  0.309   \\  
		PDS (CV)  &  0.064  & & 0.068  & &   0.068   && 0.073  \\  
		
		PMLE-DR &   0.051  & & 0.050   & & 0.063   &&  0.071  \\  
		BR-DR  & 0.055  & &  0.052   & &  0.069   && 0.041 \\   \hline 

	\end{tabular}
	\vspace*{-6pt}
\end{table}

\newpage

\begin{table}[h]
	\centering
	\caption{ Type I errors based on 1,000 replications in settings with a denser propensity score: $\Sigma=[\sigma_{i,j}]_{1\leq i,j\leq p}$.}
	\label{Table2app}
	\begin{tabular}{l  r r r r r r r}
		\\
		\textit{Correct models} &&\\
		
		\hline
		& $n=200$   &   &  $n=500$  & & $n=200$ & & $n=200$\\  
		Methods    &  $p=200$ &    &  $p=500$ & & $p=100$ & &  $p=250$ \\ \hline
		Standard na\"ive (forced) & 0.431  & & 0.368   & &  0.197   &&  0.505  \\  
		Standard na\"ive (not forced) & 0.173   & & 0.174   & & 0.088    && 0.203   \\  
		PDS (pre-specified) & 0.132   & &  0.082  & &  0.096   && 0.143   \\  
		PO (pre-specified) & 0.102  & & 0.071   & &  0.078   && 0.094   \\  
		PDS (CV) & 0.056  & &  0.071  & & 0.064    &&  0.068  \\  
		PMLE-DR & 0.050  & & 0.062   & &  0.046   &&  0.062   \\
		  
		BR-DR  &  0.047  & & 0.041   & & 0.043    && 0.041 \\  \hline \\

		\textit{Incorrect outcome model} &&\\
		\hline
		& $n=200$  & & $n=500$  && $n=200$ && $n=200$ \\ 
		Methods  &  $p=200$ &  &  $p=500$  && $p=100$ && $p=250$\\ \hline
		Standard na\"ive (forced) & 0.315  & &  0.265  & & 0.156    &&  0.353  \\  
		Standard na\"ive (not forced) & 0.126  & & 0.116   & & 0.090    &&  0.137  \\  
		PDS (pre-specified) & 0.097  & & 0.066   & & 0.070    && 0.091   \\  
		PO (pre-specified)  & 0.073  & & 0.066   & &  0.056   && 0.074   \\  
		PDS (CV)  & 0.063  & &  0.060  & &   0.060  &&  0.068  \\  
		
		PMLE-DR & 0.055  & & 0.063  & & 0.048   && 0.064    \\  
		BR-DR  & 0.038  & & 0.043   & & 0.048    && 0.037  \\   \hline 
\\
\\

	\end{tabular}
	\vspace*{-6pt}
\end{table}

\vspace{-1cm}
\begin{table}[h]
	\centering
	\caption{ Type I errors based on 1,000 replications in settings with heteroscedastic errors: $\Sigma=\textnormal{I}_{p\times p}$.}
	\label{Table3app}
\resizebox{14cm}{!}{	\begin{tabular}{l  r r r r r r r}
		\\
		\textit{Correct models} &&\\
		
		\hline
		& $n=200$   &   &  $n=500$  & & $n=200$ & & $n=200$\\  
		Methods    &  $p=200$ &    &  $p=500$ & & $p=100$ & &  $p=250$ \\ \hline
		Standard na\"ive (forced) & 0.532 && 0.815 && 0.279  && 0.613 \\  
		Standard na\"ive (not forced) & 0.260  &&  0.550 && 0.171  &&0.324  \\  
		PDS (pre-specified) & 0.543 && 0.779 && 0.549  &&0.531  \\  
		PO (pre-specified) & 0.502  && 0.738 &&0.517   && 0.497 \\  
		PDS (CV) & 0.073 && 0.060 && 0.062  && 0.079 \\  
		PMLE-DR & 0.043 && 0.060  && 0.057  && 0.077 \\  
		BR-DR  &  0.060 && 0.062 && 0.069  && 0.054 \\ \hline \\

				\textit{Incorrect outcome model} &&\\
		\hline
		& $n=200$   &   &  $n=500$  & & $n=200$ & & $n=200$\\  
		Methods    &  $p=200$ &    &  $p=500$ & & $p=100$ & &  $p=250$ \\ \hline
		Standard na\"ive (forced) & 0.375  && 0.590 &&  0.203 && 0.436 \\  
		Standard na\"ive (not forced) &0.152  && 0.309  && 0.105  && 0.199 \\  
		PDS (pre-specified) &0.348  && 0.627 && 0.369  && 0.345  \\  
		PO (pre-specified) &0.310  && 0.594 && 0.348  && 0.322 \\  
		PDS (CV) &0.066 && 0.064 && 0.071  && 0.068 \\  
		PMLE-DR & 0.045 && 0.058 && 0.055  && 0.067 \\  
		BR-DR  & 0.042 && 0.069 &&  0.072  && 0.049   \\ \hline \\  

	\end{tabular}
}
\end{table}

\bibliographystyle{apalike}
\bibliography{HDinf_MM}
   
\end{document}